\documentclass[letterpaper]{article}
\usepackage[margin=1in]{geometry}
\usepackage[T1]{fontenc}
\usepackage{amsmath}
\usepackage{amssymb}
\usepackage{amsthm}
\usepackage{xcolor}
\definecolor{darkgreen}{rgb}{0,0.5,0}
\usepackage{hyperref}
\hypersetup{
    unicode=false,
    colorlinks=true,
    linkcolor=red,
    citecolor=darkgreen,
    filecolor=magenta,
    urlcolor=blue
}
\usepackage{authblk}
\usepackage{enumerate}
\usepackage{aliascnt}

\theoremstyle{plain}
\newtheorem{theorem}{Theorem}[section]

\newaliascnt{proposition}{theorem}
\newtheorem{proposition}[proposition]{Proposition}
\aliascntresetthe{proposition}

\newaliascnt{lemma}{theorem}
\newtheorem{lemma}[lemma]{Lemma}
\aliascntresetthe{lemma}

\newaliascnt{corollary}{theorem}
\newtheorem{corollary}[corollary]{Corollary}
\aliascntresetthe{corollary}

\newaliascnt{claim}{theorem}

\aliascntresetthe{claim}

\theoremstyle{definition}
\newaliascnt{definition}{theorem}
\newtheorem{definition}[definition]{Definition}
\aliascntresetthe{definition}

\newaliascnt{assumption}{theorem}

\aliascntresetthe{assumption}

\theoremstyle{remark}
\newaliascnt{remark}{theorem}
\newtheorem{remark}[remark]{Remark}
\aliascntresetthe{remark}

\newaliascnt{conjecture}{theorem}

\aliascntresetthe{conjecture}

\usepackage[capitalize,nameinlink]{cleveref}

\newcommand{\eps}{\varepsilon}
\newcommand{\mms}{\mathrm{MMS}}
\newcommand{\Like}{\textnormal{\textsc{Like}}}
\newcommand{\Rand}{\textnormal{\textsc{Rand}}}
\newcommand{\Env}{\mathrm{Env}}

\title{\bf Closing Gaps in Online Fair Division}
\author[1]{Tzeh Yuan Neoh}
\author[2]{Nicholas Teh}
\affil[1]{Harvard University, USA}
\affil[2]{University of Oxford, UK}
\date{}

\begin{document}
\maketitle

\begin{abstract}
We study the online fair division of indivisible items, where items arrive one at a time and must be allocated immediately and irrevocably. We address three central open questions in the literature.

First, we show that for every $n\ge 2$ agents, every fixed $k\ge 1$, and every $\alpha\in(0,1]$, no online algorithm can guarantee $\alpha$-PROP$k$ against an adaptive adversary. This remains true even when the total number of goods is known in advance, all values lie in $[0,1]$, and every good is positively valued by at most two agents. The impossibility extends to a broad range of standard envy-based, proportionality-based, and share-based fairness notions. We also establish an analogous impossibility for chores.

Second, in the setting with predictions, a lightweight form of future information, the maximum item value, was previously known to guarantee only $1/n$-PROP1, leaving open whether the dependence on $n$ is necessary. We give a deterministic $19/30$-PROP1 algorithm against adaptive adversaries that does not require knowing the total number of goods. Given an additional upper bound $\kappa\in[2,n]$ on the number of agents who value any good positively, the guarantee improves to $\max \{ 19/30, n/(n+\kappa) \}$. The guarantee remains a positive constant under any fixed one-sided prediction error below one. When predictions are exact and the total number of goods $m\ge n\log n$ is known in advance, a deterministic algorithm achieves the same $19/30$-PROP1 factor together with $O(\sqrt{m\log n/n})$ maximum additive envy after normalizing each agent's values by their maximum item value.

Third, against a non-adaptive adversary, we determine the tight high-probability PROP1 guarantee of the classical \textsc{Like} rule, which assigns each good uniformly among the agents who value it positively. Its guarantee improves when fewer agents value the same good, unlike uniform random allocation.
\end{abstract}

\section{Introduction}

Fair division of indivisible goods is a fundamental problem dealing with how to allocate discrete resources among agents with different preferences in a fair manner \cite{brams1996fairdivision,moulin2003fairdivision}. 
Applications of this model are widespread in modern ML systems.
Examples include assigning ad impressions or recommendation exposure, scheduling compute jobs in shared infrastructure, and distributing content moderation tasks or service requests as they arrive \cite{aleksandrov2015onlinefoodbank,aleksandrov2020online,neoh2026online,choo2026approxproponline}. In these applications, delaying every decision until the entire sequence is known is often impossible, and revising past decisions may be infeasible.

These applications can be studied through the \emph{online fair division} model, where indivisible goods arrive sequentially and must be
allocated immediately and irrevocably to an agent based on their reported values for the current good.\footnote{For our main impossibility result, we also establish a corresponding result for \emph{chores}, modeled as items with \emph{nonnegative costs}.}
A key challenge is to provide
ex post fairness guarantees on the final allocation, even though the algorithm made every decision online, without knowing the future.
Key fairness properties studied in (both the offline and online) fair division literature include \emph{envy-freeness up to one good} (EF1), \emph{proportionality up to one good} (PROP1), and \emph{maximin share fairness} (MMS).
However, while these fairness notions are often possible to achieve exactly (or approximately) offline, they can become impossible to even approximate online.

For instance, for MMS, Zhou et al.~\cite{zhou2023icml_mms_chores} showed that in the goods setting, even if we have information about the sum of agents' values over all incoming goods (otherwise known as \emph{normalization} information), no algorithm can achieve any approximation to MMS for three or more agents, whereas a $1/2$-approximation is possible in the case of two agents.
Subsequent work studies MMS under further assumptions.
Kulkarni et al.~\cite{KulkarniMehtaShahkar2025} consider sequentially
arriving agents with known valuation types, whereas Wang and
Wei~\cite{wang2026onlinefairBinary} retain sequentially arriving items
and restrict valuations to binary and personalized two-value classes.

Another line of work asks whether information about future goods changes what can be achieved online.
Neoh et al.~\cite{neoh2026online} study normalization information, which gives each agent's total value $v_i(G)$, and frequency predictions, which give each agent the multiset of values that will appear but not their arrival order.
Without future information, they show that no deterministic online algorithm guarantees any positive approximation to EF1. With normalization information, they obtain PROP1 for every number of agents $n$ and EF1 for $n=2$, but no positive EF1 approximation is possible for $n\ge3$. With frequency predictions, they match the best-known offline guarantees for a broad class of share-based fairness notions.

In a different direction, recent work studies restricted valuation domains
such as binary or personalized two-value instances~\cite{amanatidis2025personalized2value,wang2026onlinefairBinary}.
Kahana et al.~\cite{kahana2026perpetual} study prefix-wise fairness with an additive
error that may grow over time, while Adams and Segal-Halevi~\cite{adams2026perpetuallyfairassignmentsbalanced} study ordinal PROP1 and PROP2 in a repeated assignment model based on sequences of permutations.

Benadè et al.~\cite{benade2018envyvanish} and Neoh et al.~\cite{neoh2026online} show that exact PROP1 need not exist online.
PROP1 has a long history in the offline fair division literature as a weak benchmark that is implied by EF1 and MMS, so it remains a reasonable fairness target even when stronger notions are unattainable.

This leaves a basic and conceptually important question: are the widespread negative results in the online fair division literature caused mainly by the strength of notions such as EF1, EFX, or MMS, or does the standard online model already rule out every positive approximation to PROP1 \cite{neoh2026online,choo2026approxproponline}?
PROP1 is weaker than EF1 and many share-based guarantees, and asks only that each agent reach their proportional share after adding a single good outside their bundle. Thus, a positive approximation would identify a meaningful fairness guarantee under irrevocable online allocation. Conversely, an impossibility for PROP1 would show that the difficulty comes from the standard online model itself and would help explain why positive results change the information, valuation, or adversarial assumptions.

Choo et al.~\cite{choo2026approxproponline} initiated a systematic study of this question. They showed that several natural greedy rules do not guarantee any positive multiplicative approximation to PROP1 against an adaptive adversary. Against a non-adaptive adversary, they proved that the uniform random allocation rule \Rand{}, which assigns each good uniformly among all agents, has a PROP1 approximation factor of $\Theta(1/\log(n/\delta))$ with probability at least $1-\delta$. With exact maximum item value (MIV) predictions, they obtained a deterministic $1/n$-PROP1 guarantee against an adaptive adversary.

The classical \Like{} rule \cite{aleksandrov2015onlinefoodbank} is a natural comparison with \Rand{}: unlike \Rand{}, it assigns a good only among agents who value it positively. When few agents value a good, each such agent therefore receives it with greater probability. It was not known whether this improves the ex post PROP1 guarantee.

The current literature therefore leaves three central open questions:
\begin{enumerate}
    \item Can any online algorithm guarantee a positive approximation to PROP1, or even to PROP$k$, against an adaptive adversary \cite{neoh2026online,choo2026approxproponline}?
    \item Can exact MIV predictions give a constant PROP1 approximation independent of $n$, improving on $1/n$ \cite{choo2026approxproponline}?
    \item Against a non-adaptive adversary, can the classical \Like{} rule obtain a better PROP1 guarantee than \Rand{} when few agents value the same good positively \cite{aleksandrov2015onlinefoodbank,choo2026approxproponline}?
\end{enumerate}
We answer all three questions in this paper.

\subsection{Our Contributions}
Our results distinguish the fully online setting from two settings in which positive guarantees become possible: (i) limited information about future goods and (ii) randomization against a non-adaptive adversary.

\paragraph{1. No information about future goods: a strong negative result (\Cref{sec:inapprox}).}
For every $n\ge2$, $k\ge1$, and $\alpha\in(0,1]$, no deterministic online algorithm can guarantee $\alpha$-PROP$k$ against an adaptive adversary. This remains true even when the final number of goods is known in advance, all valuations lie in $[0,1]$, and each good is positively valued by at most two agents. The same impossibility also extends to randomized algorithms.

The proof first shows that $\alpha$-PROP$k$ implies $(\alpha/k)$-PROP1 and then reduces the problem to two agents. For two agents, we track the slack in each agent's current PROP1 inequality, normalized by the value of their best good outside their bundle. The adversary repeatedly returns one agent's normalized slack to a controlled range while decreasing the other's by a fixed amount. Eventually both normalized slacks are below a fixed threshold, so one final good valued by both agents makes the agent who does not receive it violate $\alpha$-PROP1. Appendix~\Cref{app:fairness-implications} proves the implications needed to extend the negative result to standard multiplicative versions of EF, EFX, EF1, EEFX, MXS, MMS, PMMS, GMMS, PROPX, PROPm, PROPavg, Avg-EFX, and every fixed EF$c$ notion considered there.

Appendix~\ref{app:chores} proves the corresponding impossibility for chores: for every fixed $k$, no online algorithm guarantees $\lambda$-PROP$k$ for any $\lambda<n$. The appendix also derives impossibilities for standard proportionality-based, envy-based, and share-based chore notions.

This impossibility also clarifies the role of assumptions used in positive results. Such results provide information about future goods \cite{zhou2023icml_mms_chores,neoh2026online,choo2026approxproponline}, restrict the valuation domain \cite{amanatidis2025personalized2value,wang2026onlinefairBinary}, use randomization against a non-adaptive adversary \cite{choo2026approxproponline}, allow limited buffering or reordering \cite{amanatidis2026buffers}, require fairness after every prefix \cite{adams2026perpetuallyfairassignmentsbalanced,kahana2026perpetual}, or compare online fairness with the best fairness attainable offline on the same instance \cite{chen2026competitiveanalysis}.

\paragraph{2. Constant approximation to PROP1 with maximum item value predictions (\Cref{sec:miv}).}
For each agent $i$, a perfect maximum item value (MIV) prediction gives the single number $p_i=\max_{g\in G}v_i(g)$; it reveals neither $v_i(G)$ nor the individual values or arrival order of future goods. Thus it provides only one scalar per agent. Choo et al. \cite{choo2026approxproponline} obtained $1/n$-PROP1 from these predictions and showed that stronger notions such as EF1, MMS, and PROPX still have no positive approximation. We give a deterministic algorithm that guarantees $19/30$-PROP1 against adaptive adversaries, without knowing the horizon. If the algorithm also knows that at most $\kappa$ agents value any one good positively, we obtain $\max\{19/30,n/(n+\kappa)\}$-PROP1 by choosing between this algorithm and a reciprocal-potential rule.

For each agent, the prediction accounts for the one outside good allowed by PROP1. We measure the resulting slack relative to the target proportional value and normalize it by the prediction. The key to the $19/30$ factor is to compare the total potential change across possible recipients: one recipient's decrease must compensate for the increases from all agents who do not receive the good. We bound these increases both individually and relative to the current total potential. With a different potential function, this comparison gives $19/30$ without assigning a separate receiving probability to every agent. A separate reciprocal-potential argument gives $n/(n+\kappa)$.

Using the error transformation of Choo et al.~\cite{choo2026approxproponline}, the $19/30$ guarantee becomes $\frac{19n(1-\eps)}{30n-19\eps}$-PROP1 under one-sided error $\eps$. When $\kappa$ is given, we can instead use $\frac{n(1-\eps)}{n+\kappa-\eps}$ whenever it is larger. The resulting factor decreases continuously with $\eps$ and remains a positive constant for every fixed $\eps<1$.

With exact MIV predictions and a known horizon $m$, we can also control envy without reducing the PROP1 factor. For $m\ge n\log n$, a deterministic algorithm simultaneously guarantees $19/30$-PROP1 and $O(\sqrt{m\log n/n})$ maximum additive envy after each agent's values are divided by their predicted maximum item value. The hidden constant is absolute.

The proof combines the same PROP1 potential with the exponential potential method for pairwise envy of Benad\`e et al.~\cite{benade2018envyvanish}. Each envy term is multiplied by a term depending on the envied agent's PROP1 slack. The key is to compare the combined change across recipients rather than require both objectives to be preserved separately: either uniform averaging controls both potentials, or one recipient's decrease compensates for their combined increase. This allows the $19/30$ factor and the envy bound to hold for the same allocation.

Appendix~\ref{app:miv-combined} gives a shorter combined potential argument whose PROP1 factor is $\Theta(1/\log n)$, together with the explicit normalized envy bound $11\sqrt{m\log n/n}$. We include it for its simpler proof idea.

\paragraph{3. Random allocation against a non-adaptive adversary: matching PROP1 bounds and an envy comparison (\Cref{sec:rand}).}
We compare two basic independent randomized rules to analyze the effect of assigning a good only among agents who value it positively. \Rand{} assigns every good uniformly among all $n$ agents \cite{choo2026approxproponline}, whereas \Like{} assigns it uniformly among the agents who value it positively \cite{aleksandrov2015onlinefoodbank}. Let $\kappa:=\max_{g\in G}|\{i\in N:v_i(g)>0\}|$.
For every $\delta\in(0,1/2]$, the worst-case factor guaranteed by \Like{} with probability at least $1-\delta$ is $\Theta(\min \{1,\frac{n}{\kappa\log(n/\delta)} \})$.
Both the upper and lower bounds hold for every $\kappa$, and the lower bound uses binary valuations. In contrast, \Rand{} has worst-case factor $\Theta(1/\log(n/\delta))$ even when each good is positively valued by exactly one agent. Thus, assigning a good only among agents who value it can improve proportionality by a factor of order $n/\kappa$.

For additive envy, we prove instance-specific concentration bounds for both rules. The bound for \Like{} depends on $\sum_g v_i(g)^2/|N(g)|$, while the corresponding quantity for \Rand{} is $n^{-1}\sum_g v_i(g)^2$. The latter is never larger. We also give binary instances on which the expected maximum envy of \Like{} is $\Omega(\sqrt{m/\kappa})$. These bounds show a tradeoff between the two rules: restricting recipients improves \Like{}'s proportionality when $\kappa$ is small, but it also increases the variance relevant to additive envy, and our binary lower bound shows that this increase can occur in the actual envy.

\subsection{Related Work}

\paragraph{Online fair division.}
Aleksandrov et al.~\cite{aleksandrov2015onlinefoodbank} introduced the food bank model and studied the fairness and incentive properties of simple online rules, including \Like{}. Under truthful reports, \Like{} is envy-free ex ante, but its ex-post envy is unbounded even for binary valuations and two agents. Benad\`e et al.~\cite{benade2018envyvanish} studied additive envy over time and designed an online allocation rule with vanishing envy per good.
Aleksandrov and Walsh~\cite{aleksandrov2020online} survey the earlier literature.

More recent work studies fairness of the final allocation. Zhou et al. \cite{zhou2023icml_mms_chores} study MMS under normalization information. Neoh et al. \cite{neoh2026online} compare several kinds of information about future goods and give positive and negative results for EF, proportionality, and MMS. Choo et al. \cite{choo2026approxproponline} focus on approximate PROP1, including greedy lower bounds, uniform random allocation against a non-adaptive adversary, and maximum item value predictions. 
Melissourgos and Protopapas~\cite{melissourgos2026onlineefx} study online EFX with predictions of the agents' full valuation vectors and measure prediction error by total variation distance. Their prediction input and fairness objective differ from the MIV prediction input and PROP1 objective studied here.

Positive results are also known for binary and two-value valuation classes \cite{amanatidis2025personalized2value,wang2026onlinefairBinary}. For chores, Seddighin and Seddighin~\cite{seddighin2025lowerbound} and Song et al.~\cite{song2025onlinemmschores} give contemporaneous $n$-factor lower bounds for online MMS. Their constructions use recursion over the number of agents and rapidly increasing costs.
Appendix~\ref{app:chores} strengthens the inductive statement to control the cost left after removing any fixed number $k$ of chores.
This gives the PROP$k$ lower bound, applies to randomized algorithms against adaptive adversaries, and gives consequences for several other notions.

Our first result differs from the earlier negative results in two ways. It applies to every online algorithm rather than selected greedy rules, and it rules out every positive approximation even after each agent may add any fixed number of goods. Our maximum item value result improves the previous $1/n$ PROP1 factor to $19/30$ without asking for more information.

\paragraph{Other online models.}
Some papers ask for fairness after every prefix, with an additive error that may grow with time \cite{kahana2026perpetual}. Adams and Segal-Halevi~\cite{adams2026perpetuallyfairassignmentsbalanced} study repeated assignments generated by sequences of permutations and show that certain balance conditions imply ordinal PROP1 or PROP2 after every day. Amanatidis et al.~\cite{amanatidis2026buffers} allow a bounded number of goods to be held temporarily and reordered before allocation. These models use different arrival or fairness requirements and
are therefore complementary to our negative result. Cohen et al. \cite{cohen2026budgetconstraints} study a different extension in which each assignment must satisfy a budget constraint and some goods may be left unallocated.

A related model is \emph{temporal fair division}, which evaluates the cumulative allocation after every prefix. Elkind et al. \cite{elkind2025temporal} assume that the complete arrival sequence is known in advance and study temporal EF1 for goods and chores. Choi et al. \cite{choi2026tfdmm} extend this model to mixed manna and identify restricted settings in which temporal fairness can be achieved by online rules. Goldberg et al.~\cite{Goldberg2026} study computational questions related to minimizing the sum of envy in each round. These works require fairness throughout the sequence, whereas our guarantees concern the final allocation and, unless stated otherwise, use no information about future goods.

Chen and Tan \cite{chen2026competitiveanalysis} give a broad competitive analysis for several fairness notions, goods and chores, and several valuation assumptions. Their ratio compares the fairness achieved online with the best fairness level attainable offline on the same instance. Our approximation factor instead compares each agent's final value directly with the PROP$k$ requirement, as most other prior works in the area do. The two comparisons answer different questions.

For envy minimization, Halpern et al. \cite{halpern2025onlineenvy} determine the optimal order for general algorithms against a non-adaptive adversary through a connection with online multicolor discrepancy. Our \Cref{sec:rand} instead compares \Like{} and \Rand{} on the same fixed instance to determine the effect of assigning each good only among agents who value it positively. The comparison shows that this choice can improve proportionality while increasing additive envy, a question not answered by the optimal worst-case envy rate for general algorithms.

\section{Preliminaries}
\label{sec:preliminaries}

For a positive integer $z$, let $[z]:=\{1,\dots,z\}$. There are $n\ge2$ \emph{agents} $N=[n]$ and a finite sequence of $m\ge1$ indivisible \emph{goods} $G=\{g_1,\dots,g_m\}$, indexed by arrival order. 
The algorithm need not know $m$ (which we sometimes call the \emph{horizon}) in advance unless stated otherwise. 
Each agent $i\in N$ has a nonnegative additive valuation $v_i:2^G\to\mathbb{R}_{\ge0}$, so $v_i(S)=\sum_{g\in S}v_i(g)$ for every $S\subseteq G$. 
We call an instance \emph{binary} if $v_i(g)\in\{0,1\}$ for every $i\in N$ and $g\in G$.
A \emph{(complete) allocation} is a tuple $A=(A_1,\dots,A_n)$ that partitions $G$.

For a good $g$, let $N(g):=\{i\in N:v_i(g)>0\}$ and $\kappa:=\max_{g\in G}|N(g)|$.
Thus, $\kappa$ is the largest number of agents who value the same good positively. We use this parameter only in results that state it explicitly.

\begin{definition}[$\alpha$-PROP$k$]
    Fix $\alpha\in[0,1]$ and an integer $k\ge1$. An allocation $A$ is $\alpha$-\emph{proportional up to $k$ goods} ($\alpha$-PROP$k$) if, for every agent $i$, there is a set $S_i\subseteq G\setminus A_i$ with $|S_i|\le k$ such that $v_i(A_i\cup S_i)\ge \alpha \cdot \frac{v_i(G)}{n}$.
\end{definition}
When $k=1$, we write $\alpha$-PROP1; when $\alpha=1$, we omit $\alpha$. In particular, PROP1 is the usual exact notion. For $k=1$, the definition is equivalent to $v_i(A_i)+\max_{g\in G\setminus A_i}v_i(g) \ge \alpha \cdot \frac{v_i(G)}{n}$,
where, throughout, the maximum over an empty set is understood to be zero.

For an allocation $A$, its \emph{maximum additive envy} is $\Env(A):=\max_{i,j\in N}(v_i(A_j)-v_i(A_i))$.
Appendix~\ref{app:fairness-implications} defines the other goods fairness notions mentioned in the paper and proves the implications used for our goods lower bounds. Appendix~\ref{app:chores} gives the corresponding definitions and lower bounds for chores.

\paragraph{Online algorithms and adversaries.}
At time $t$, good $g_t$ arrives and the vector $(v_1(g_t),\dots,v_n(g_t))$ is revealed. The algorithm must assign $g_t$ immediately to one agent. Every good is assigned exactly once, and earlier assignments cannot be changed. Fairness is evaluated on the final allocation.
A deterministic algorithm has no internal randomness. A randomized algorithm may use private random choices. An \emph{adaptive adversary} chooses each next valuation vector after seeing the earlier goods and their realized recipients; it does not see future private random choices. Unless the horizon is known in advance, it may also decide when the finite sequence ends. A \emph{non-adaptive adversary} fixes the complete valuation sequence before the algorithm makes any random choice.
Against an adaptive adversary, a randomized algorithm guarantees a property with probability at least $1-\delta$ if the property holds with that probability against every allowed adversarial strategy. On a fixed instance, probability is only over the algorithm's random choices.

\paragraph{Maximum item value (MIV) predictions.}
For each agent $i$, let $p_i^{\max}:=\max_{g\in G}v_i(g)$.
An exact MIV prediction gives the algorithm $p_i=p_i^{\max}$ before any good arrives. It gives only one number per agent and reveals neither $v_i(G)$ nor the individual values or arrival order of future goods. Following Choo et al.~\cite{choo2026approxproponline}, predictions have one-sided error $\eps\in[0,1)$ if $p_i^{\max}\in[(1-\eps)p_i,p_i]$ for every agent $i$. Thus each prediction is an upper bound on the true maximum and overestimates it by at most the factor $1/(1-\eps)$. The upper-bound direction ensures that $v_i(g)/p_i\le1$ for every agent with $p_i>0$ and every good $g$. Choo et al. also observe that a multiplicative two-sided prediction can be converted to this form by scaling it upward.

\section{Inapproximability of \texorpdfstring{PROP$k$}{PROPk} in the Fully Online Setting}
\label{sec:inapprox}

We now show that no positive multiplicative approximation to PROP$k$ is possible in the standard fully online model.
We state the result for deterministic algorithms; against an adaptive adversary, randomization offers no additional power, so the same impossibility extends to randomized algorithms.

\begin{theorem}
\label{thm:online-propk-impossibility}
Fix $n\ge2$, $k\ge1$, and $\alpha\in(0,1]$. For every deterministic online algorithm, there is an adaptive adversary under which the final allocation is not $\alpha$-PROP$k$. The statement remains true even when all of the following hold:
\begin{enumerate}[(i)]
    \item the final number of goods is known before the first good arrives;
    \item every value lies in $[0,1]$; and
    \item every good is positively valued by at most two agents.
\end{enumerate}
\end{theorem}

The proof has three steps. First, allowing $k$ added goods can improve the factor by at most $k$. Second, any guarantee for $n$ agents gives a guarantee for two agents. We then construct the two-agent adversary.

\begin{lemma}
\label{lem:propk-implies-prop1}
Under nonnegative additive valuations, every $\alpha$-PROP$k$ allocation is also $(\alpha/k)$-PROP1.
\end{lemma}

\begin{proof}
Fix an agent $i$, and define $h_i:=\max_{g\in G\setminus A_i}v_i(g)$,
where the maximum is zero if $A_i=G$. Since the allocation is $\alpha$-PROP$k$, there is a set $S_i\subseteq G\setminus A_i$ with $|S_i|\le k$ such that $v_i(A_i)+v_i(S_i)
    \ge
    \alpha\frac{v_i(G)}{n}$.
Every good in $S_i$ has value at most $h_i$, so $v_i(S_i)\le kh_i$. Thus,
\[
    k(v_i(A_i)+h_i)
    \ge
    v_i(A_i)+kh_i
    \ge
    v_i(A_i)+v_i(S_i)
    \ge
    \alpha\frac{v_i(G)}{n},
\]
where the first inequality uses $k\ge1$ and $v_i(A_i)\ge0$. Therefore,
\[
    v_i(A_i)+h_i
    \ge
    \frac{\alpha}{k}\frac{v_i(G)}{n}.
\]
By the equivalent condition for PROP1 stated in \Cref{sec:preliminaries}, agent $i$ satisfies $(\alpha/k)$-PROP1. Since $i$ was arbitrary, the result follows.
\end{proof}

\begin{lemma}
\label{lem:two-agents-suffice}
Fix $n\ge3$ and $\alpha\in(0,1]$. If a deterministic online algorithm guarantees $\alpha$-PROP1 for $n$ agents, then a deterministic two-agent online algorithm guarantees $(2\alpha/n)$-PROP1 against the same class of adversaries.
\end{lemma}

\begin{proof}
Let $\mathcal A$ be the $n$-agent algorithm. We construct a two-agent algorithm $\mathcal B$. When a good $g$ arrives with values $(v_1(g),v_2(g))$, algorithm $\mathcal B$ runs $\mathcal A$ on the value vector $(v_1(g),v_2(g),0,\dots,0)$. If $\mathcal A$ assigns $g$ to agent $2$, then $\mathcal B$ also assigns it to agent $2$; otherwise, $\mathcal B$ assigns it to agent $1$. 

The adversary class is preserved. A fixed two-agent sequence remains fixed after appending zeros. Given an adaptive two-agent adversary, the corresponding $n$-agent adversary appends zeros to each value vector and, after $\mathcal A$ assigns the good, maps recipient $2$ to agent $2$ and every other recipient to agent $1$ before generating the next vector.

Let $(A_1,\dots,A_n)$ be the allocation produced by $\mathcal A$, and let $D:=A_3\cup\cdots\cup A_n$. By construction, $\mathcal B$ produces the bundles $B_1:=A_1\cup D$ and $B_2:=A_2$.
For each $i\in\{1,2\}$, the $\alpha$-PROP1 guarantee of $\mathcal A$ gives
\begin{equation}
\label{eq:internal-prop1}
    v_i(A_i)+\max_{g\in G\setminus A_i}v_i(g)
    \ge
    \frac{\alpha}{n}v_i(G).
\end{equation}
For agent $1$ under $\mathcal B$:
\begin{equation} \label{eqn:ext1}
    v_1(B_1)+\max_{g\in G\setminus B_1}v_1(g) = v_1(A_1)+v_1(D)+\max_{g\in A_2}v_1(g) \ge v_1(A_1)+\max_{g\in G\setminus A_1}v_1(g).
\end{equation}
The inequality uses $G\setminus A_1=A_2\cup D$ and $v_1(D)\ge\max_{g\in D}v_1(g)$, which follows from nonnegativity and additivity.
For agent $2$ under $\mathcal B$:
\begin{equation} \label{eqn:ext2}
    v_2(B_2)+\max_{g\in G\setminus B_2}v_2(g)
    =v_2(A_2)+\max_{g\in G\setminus A_2}v_2(g).
\end{equation}

Combining \eqref{eqn:ext1} and \eqref{eqn:ext2} with \eqref{eq:internal-prop1}, each agent $i\in\{1,2\}$ satisfies
\[
    v_i(B_i)+\max_{g\in G\setminus B_i}v_i(g)
    \ge
    \frac{\alpha}{n}v_i(G)
    =
    \frac{2\alpha}{n}\frac{v_i(G)}{2}.
\]
Thus $(B_1,B_2)$ is $(2\alpha/n)$-PROP1.
\end{proof}
Choo et al.~\cite{choo2026approxproponline} use the current PROP1 ratio to analyze selected greedy rules. The construction below instead derives update formulas for the normalized slack in each agent's current PROP1 inequality and applies to every possible sequence of recipients. This is what allows the lower bound to cover every online algorithm.

It remains to prove the two-agent statement. The next lemma applies to every possible sequence of recipients.

\begin{lemma}
\label{lem:two-agent-prop1-impossibility}
Fix $\alpha\in(0,1]$. There is an adaptive strategy for two agents such that every possible sequence of recipients reaches, after finitely many goods, a prefix that is not $\alpha$-PROP1.
\end{lemma}

\begin{proof}
Let $G'$ be the goods seen so far, and let $A_1,A_2$ be the current bundles. Set $\gamma:=\frac{\alpha}{2}$ and $\beta:=1-\gamma$.
For agent $i\in\{1,2\}$, let $m_i:=\max_{g\in A_{-i}}v_i(g)$,
where $A_{-i}$ is the other agent's bundle and the maximum is zero when that bundle is empty. The current allocation is $\alpha$-PROP1 exactly when
\begin{equation}
\label{eq:two-agent-prop1}
    v_i(A_i)+m_i\ge \gamma v_i(G')
    \quad \text{for } i=1,2.
\end{equation}

We first make both $m_i$ positive. Present goods with values $(1,1)$. If one agent has received none after $\ell$ such goods, then that agent has $v_i(A_i)=0$, $m_i=1$, and $v_i(G')=\ell$. Inequality \eqref{eq:two-agent-prop1} is violated once $\ell>1/\gamma$. Thus, after finitely many goods, either the allocation is not $\alpha$-PROP1 or both agents have received a positively valued good. From then on, $m_1,m_2>0$.

For the rest of the proof, define the normalized slack in agent $i$'s PROP1 inequality
\begin{equation*}
    r_i:=\frac{v_i(A_i)+m_i-\gamma v_i(G')}{m_i}.
\end{equation*}
The current allocation satisfies $\alpha$-PROP1 exactly when $r_1,r_2\ge0$. Suppose the next good has values $(u_1m_1,u_2m_2)$, where the $m_i$ are measured immediately before the good arrives. The update for agent $i$ is as follows:
\begin{equation}
\label{eq:margin-updates}
    r_i^+=
    \begin{cases}
        r_i+\beta u_i,
        &\text{if agent $i$ receives the good},\\[1mm]
        r_i-\gamma u_i,
        &\text{if agent $i$ does not receive it and $u_i\le1$},\\[1mm]
        \displaystyle \beta+\frac{r_i-1}{u_i},
        &\text{if agent $i$ does not receive it and $u_i>1$}.
    \end{cases}
\end{equation}
The first two lines follow because $m_i$ does not change. In the third line, the new good becomes the best good in the other bundle, so the new denominator is $u_i m_i$.

In particular, if
\begin{equation}
\label{eq:both-margins-small}
    r_1<\gamma
    \quad\text{and}\quad
    r_2<\gamma,
\end{equation}
then one more good with normalized values $(1,1)$ makes the resulting allocation not $\alpha$-PROP1: whichever agent does not receive it has $r_i^+=r_i-\gamma<0$.

Let $q:=\lfloor\beta/\gamma\rfloor$ and $\eta:=\frac{\beta+(q+1)\gamma}{2}$.
Since $\gamma\le1/2$, we have $q\ge1$ and
\begin{equation}
\label{eq:q-interval}
    q\gamma\le\beta<\eta<(q+1)\gamma.
\end{equation}
We next choose $t,s_1,\dots,s_q\in(0,1)$ so that
\begin{align}
\label{eq:s-conditions-a}
    \gamma s_j
    &>
    \beta\sum_{h<j}s_h+\beta t \quad \text{ for all } j=1,\dots,q,\\
\label{eq:s-conditions-b}
    \gamma
    &>
    \beta\sum_{h=1}^q s_h+\beta t.
\end{align}
The first inequality will be used when $P_j$ is the first good in the list below assigned to agent $2$, and the second when all $q$ goods are assigned to agent $1$.
To see that such a choice exists, set $R:=1+4\beta/\gamma$. Choose $c>0$ such that $cR^{q-1}<1$ and $\beta c\sum_{h=0}^{q-1}R^h<\gamma/2$, and set $s_j:=cR^{j-1}$ for all $j=1,\dots,q$, and $t:=\frac{\gamma c}{4\beta}$.

Then $t,s_1,\dots,s_q\in(0,1)$. For every $j$,
\[
    \beta\sum_{h<j}s_h<\frac{\gamma}{4}s_j,
    \quad
    \beta t=\frac{\gamma c}{4}\le\frac{\gamma}{4}s_j,
\]
which gives \eqref{eq:s-conditions-a}. Moreover, $\beta\sum_{h=1}^q s_h<\frac{\gamma}{2}$ and $\beta t<\frac{\gamma}{4}$,
so \eqref{eq:s-conditions-b} also holds.

Now, define $S_j:=\sum_{h=1}^j s_h$ and $S_0:=0$.
Start from any allocation that still satisfies $\alpha$-PROP1. Present a good with normalized values $(t,U)$, where $U>1$ is chosen from the current $r_2$. If $r_2\le1$, choose $U>\max \{1,\frac{1-r_2}{\beta}\}$; if $r_2>1$, choose $U>\max\{1,\frac{r_2-1}{\eta-\beta}\}$.
If agent $1$ receives this good, then the third line of \eqref{eq:margin-updates} gives $0<r_2^+<\eta$, and $r_1$ increases by $\beta t$. If agent $2$ receives it, then $r_1$ decreases by $\gamma t$. Repeat with a newly chosen $U$ until either the allocation is not $\alpha$-PROP1 or agent $1$ receives one of these goods. The repetition is finite: each consecutive assignment to agent $2$ decreases $r_1$ by $\gamma t>0$, so $r_1$ would eventually become negative. Thus, as long as the allocation remains $\alpha$-PROP1, the adversary reaches an allocation with $r_2\le\eta$. Over the whole repetition, $r_1$ increases by at most $\beta t$.

Suppose now that $r_2\le\eta$. Present, in order, goods $P_1,\dots,P_q$ with normalized values $P_j=(s_j,1)$, and stop the list when one is first assigned to agent $2$.
Suppose $P_j$ is the first such good. The earlier goods $P_1,\dots,P_{j-1}$ went to agent $1$, so the change in $r_1$ is $\beta S_{j-1}-\gamma s_j< -\beta t$
by \eqref{eq:s-conditions-a}. Bringing $r_2$ below $\eta$ again can increase $r_1$ by at most $\beta t$. Hence $r_1$ decreases overall.

Suppose instead that all $P_1,\dots,P_q$ go to agent $1$. As long as the allocation remains $\alpha$-PROP1, agent $2$ misses $q$ goods with second normalized value one, and therefore $r_2\le\eta-q\gamma<\gamma$ by definition of $q$, $\eta$, and \eqref{eq:q-interval}. Present one more good with normalized values $(1,1)$. If it goes to agent $1$, the resulting allocation is not $\alpha$-PROP1. Otherwise it goes to agent $2$, and the total change in $r_1$ over this list and the additional good is $\beta S_q-\gamma< -\beta t$
by \eqref{eq:s-conditions-b}. Bringing $r_2$ below $\eta$ again increases $r_1$ by at most $\beta t$.

Define
\begin{equation}
\label{eq:d-choice}
    d:=
    \min\left\{
        \min_{1\le j\le q}(\gamma s_j-\beta S_{j-1}-\beta t),
        \gamma-\beta S_q-\beta t
    \right\}.
\end{equation}
Then, \eqref{eq:s-conditions-a} and \eqref{eq:s-conditions-b} give $d>0$. We have shown that, whenever $r_2\le\eta$, the adversary either makes the allocation not $\alpha$-PROP1 or returns to an allocation with $r_2\le\eta$ and lowers $r_1$ by at least $d$.

Now, first bring $r_2$ below $\eta$. Repeat the preceding decrease until either the allocation is not $\alpha$-PROP1 or
\begin{equation}
\label{eq:r1-final-range}
    r_1<\gamma-\beta S_q.
\end{equation}
This takes finitely many repetitions because every allocation that still satisfies $\alpha$-PROP1 has $r_1\ge0$, while each completed repetition lowers $r_1$ by at least $d$.

Present $P_1,\dots,P_q$ again. If one of these goods goes to agent $2$, bring $r_2$ below $\eta$ again. Then $r_1$ decreases by at least $d$, and \eqref{eq:r1-final-range} remains true. This outcome cannot occur indefinitely without making $r_1$ negative. Therefore, if the allocation remains $\alpha$-PROP1, at some repetition all $P_1,\dots,P_q$ go to agent $1$. At that point, $r_1^+<\gamma$ and $r_2^+\le\eta-q\gamma<\gamma$.
A final good with normalized values $(1,1)$ makes the resulting allocation not $\alpha$-PROP1, by \eqref{eq:both-margins-small}.
\end{proof}

\begin{proof}[Proof of \Cref{thm:online-propk-impossibility}]
Suppose, for contradiction, that a deterministic $n$-agent online algorithm guarantees $\alpha$-PROP$k$. By \Cref{lem:propk-implies-prop1}, it would guarantee $(\alpha/k)$-PROP1. If $n=2$, applying \Cref{lem:two-agent-prop1-impossibility} with factor $\alpha/k$ gives a contradiction. If $n\ge3$, \Cref{lem:two-agents-suffice} would give a two-agent algorithm guaranteeing $(2\alpha/(kn))$-PROP1, and applying \Cref{lem:two-agent-prop1-impossibility} with this factor again gives a contradiction. The reduction presents value vectors of the form $(v_1(g),v_2(g),0,\dots,0)$, so each good is positively valued by at most two agents.

It remains to justify the known horizon and the value bound. Fix a deterministic choice whenever the adversarial strategy permits several values, and apply the strategy to every possible recipient history, stopping each branch when the allocation is no longer $\alpha$-PROP$k$. This defines a rooted tree independently of the algorithm and of any horizon announced to it. Every nonterminal node has at most $n$ children, one for each possible recipient of the next good, and the proof above shows that every branch is finite. Hence there is a common finite upper bound $M$ on all branch lengths; otherwise, K\"onig's lemma would give an infinite branch.

The adversary sets the horizon to $M$ and reveals it before the first good arrives. If a branch stops earlier, it fills the remaining positions with goods of value zero to every agent. These goods do not change any PROP$k$ inequality, so the final allocation is still not $\alpha$-PROP$k$. The resulting tree contains only finitely many values. Multiplying every value in the construction by one common positive constant makes all values lie in $[0,1]$ without changing any allocation inequality or any normalized update along a given recipient history. The algorithm may choose different recipients after this multiplication or after learning $M$, but every such history is already included in the tree. This proves all parts of the theorem.
\end{proof}

\paragraph{Randomized algorithms.}
For every randomized online algorithm, there is an adaptive adversary under which the final allocation is not $\alpha$-PROP$k$ with probability one, and items~(i)--(iii) of \Cref{thm:online-propk-impossibility} continue to hold. The construction in \Cref{lem:two-agents-suffice} is applied to the realized recipient of each good, while the strategy in \Cref{lem:two-agent-prop1-impossibility} uses only earlier realized recipients and reaches a violation for every possible sequence of recipients. Together with \Cref{lem:propk-implies-prop1}, the violation occurs for every realization of the algorithm's private random choices, and hence with probability one, by definition.

\paragraph{Consequences for other fairness notions.}
Appendix~\ref{app:fairness-implications} shows that positive multiplicative guarantees for the standard notions listed there would imply a positive $\beta$-PROP$k$ guarantee for some fixed $k$ and $\beta>0$. Hence \Cref{thm:online-propk-impossibility} also rules out those guarantees against an adaptive adversary; specifically, PROP, PROP1, PROPX, PROPm, PROPavg, Avg-EFX, EF, EFX, EF1, EEFX, MXS, MMS, PMMS, GMMS, and every fixed EF$c$ notion \cite{amanatidis2018ef,amanatidis2023fairdivisionsurvey,aziz2020prop1,BaklanovGarimidiGkatzelisSchoepflin2021AAAI,BaklanovGarimidiGkatzelisSchoepflin2021IJCAI,BarmanBiswasKrishnamurthyNarahari2018,Budish2011,caragiannis2019unreasonable,CaragiannisGargRathiSharmaVarricchio2023,conitzer2017fairpublic,KobayashiMahara2025,lipton2004ece}.\footnote{This is not an exhaustive list of the notions in the literature covered by the impossibility result.}

Appendix~\ref{app:chores} proves the corresponding PROP$k$ lower bound for chores and derives its consequences for standard chore fairness notions.

\section{Constant PROP1 Approximation with MIV Predictions}
\label{sec:miv}

Choo et al.~\cite{choo2026approxproponline} show that perfect maximum item value (MIV) predictions give $1/n$-PROP1 and leave open the optimal dependence on $n$. We remove that dependence with a $19/30$ guarantee, while retaining the factor $n/(n+\kappa)$ when it is larger. For a known horizon, we retain $19/30$-PROP1 together with normalized maximum additive envy $O(\log n+\sqrt{m\log n/n})$; when $m\ge n\log n$, this becomes $O(\sqrt{m\log n/n})$, with an absolute hidden constant.
Agents with $p_i=0$ value every good at zero, so their fairness inequalities hold automatically. Accordingly, the definitions below include only agents with $p_i>0$, while $n$ continues to denote the original number of agents.

\subsection{A Constant Factor PROP1 Approximation Independent of \texorpdfstring{$n$}{n}}
\label{subsec:miv-constant}

For an agent with $p_i>0$, let $t_i^*:=\min\{t\in[m]:v_i(g_t)=p_i\}$ be the first time a maximum-valued good appears. For $t=0,\dots,m$, let $G^t:=\{g_1,\dots,g_t\}$ and let $A_i^t$ be agent $i$'s bundle after the first $t$ goods. Define
\begin{equation}
\label{eq:miv-accounting}
    D_i^t:=v_i(G^t)+p_i\mathbf 1[t<t_i^*],
    \quad
    B_i^t:=v_i(A_i^t\setminus\{g_{t_i^*}\}).
\end{equation}
The quantity $D_i^t$ counts the predicted maximum before it arrives. At time $t_i^*$, the prediction is replaced by the actual value $p_i$, so $D_i^t$ does not change. The quantity $B_i^t$ is the value received by agent $i$, except that it does not count this first maximum-valued good. Both quantities can be updated online.

At the end, $D_i^m=v_i(G)$, and $p_i+B_i^m$ is a value permitted by PROP1. If agent $i$ did not receive $g_{t_i^*}$, they can add it and obtain $v_i(A_i)+p_i=B_i^m+p_i$. If they did receive it, then $B_i^m+p_i=v_i(A_i)$.

For a target coefficient $\rho\in(0,1/n)$, define the normalized PROP1 slack
\begin{equation}
\label{eq:miv-margin}
    s_i^t:=\frac{p_i+B_i^t-\rho D_i^t}{p_i}.
\end{equation}
Initially, $s_i^0=1-\rho$. By the preceding interpretation of $p_i+B_i^m$, positive final slacks imply $n\rho$-PROP1. For the current good $g_t$, write
\begin{equation}
\label{eq:miv-normalized-increment}
    z_i:=
    \begin{cases}
        v_i(g_t)/p_i, & p_i>0\text{ and }t\ne t_i^*,\\
        0, & \text{otherwise}.
    \end{cases}
\end{equation}
The first maximum-valued good can be recognized using $p_i$ and whether such a good has already appeared. For every agent with $p_i>0$, assigning $g_t$ to $j$ gives us
\begin{equation}
\label{eq:miv-slack-update}
    s_i^t=s_i^{t-1}-\rho z_i+z_i\mathbf 1[i=j].
\end{equation}
The treatment of the first maximum-valued good follows Choo et al.~\cite{choo2026approxproponline}.

\paragraph{From individual terms to the total potential.}
We obtain the $19/30$-PROP1 guarantee by showing that, for some recipient, the decrease in that agent's potential term offsets the total increase in the other agents' terms.

Let $f$ be the fixed positive decreasing function defined in Appendix~\ref{app:miv-function}. In particular,
\begin{equation}
\label{eq:miv-function-bounds}
    f(x)>\frac1x\quad \text{for } 0<x\le1, \quad \text{and}
    \quad
    1<f(1)<\frac{30}{19}.
\end{equation}
Set $\rho:=19/(30n)$ and define
\begin{equation}
\label{eq:miv-improved-potential}
    \Phi^t:=\sum_{i:p_i>0}f(s_i^t+\rho).
\end{equation}
The shift by $\rho$ makes every initial argument equal to one, so $\Phi^0\le nf(1)<1/\rho$. Once all arguments of $f$ are positive, this bound forces every slack to be positive. Indeed, if $s_i^t\le0$, then $0<s_i^t+\rho\le\rho<1$, and \eqref{eq:miv-function-bounds} gives us $f(s_i^t+\rho)>\frac1{s_i^t+\rho}\ge\frac1\rho$,
contradicting $\Phi^t\le nf(1)<1/\rho$. Thus it is enough to keep all arguments positive and maintain $\Phi^t\le nf(1)$. Our function behaves as $1/x$ near zero and is logarithmic on an interval containing one. This lowers $f(1)$ while preserving the bounds used to compare the possible recipients.

\paragraph{Allocation rule.}
When $g_t$ arrives, assign it to an agent maximizing
\begin{equation}
\label{eq:miv-improved-score}
    \Delta_i:=f(s_i^{t-1}+\rho-\rho z_i)
       -f(s_i^{t-1}+\rho+(1-\rho)z_i),
\end{equation}
with $\Delta_i:=0$ when $p_i=0$ and a fixed order for breaking ties. This is exactly the recipient minimizing the next value of $\Phi^t$.

We also retain the reciprocal rule for instances in which few agents value the same good. With $\rho:=1/(n+\kappa)$, this rule minimizes
\begin{equation}
\label{eq:reciprocal-potential}
    \Psi^t:=\sum_{i:p_i>0}\frac1{s_i^t}.
\end{equation}
The two rules are alternatives: when $\kappa$ is given, choose the one with the larger guarantee before the first good arrives.

Then, our result is as follows.

\begin{theorem}
\label{thm:miv-constant}
Suppose exact MIV predictions are available. Against adaptive adversaries:
\begin{enumerate}[(i)]
    \item The rule in \eqref{eq:miv-improved-score} returns a $19/30$-PROP1 allocation.
    \item If the algorithm is given an upper bound $\kappa\in[2,n]$ on $|N(g)|$ for every good, the rule in \eqref{eq:reciprocal-potential} returns an $n/(n+\kappa)$-PROP1 allocation. Choosing between the two rules above gives a $\max\{19/30,n/(n+\kappa)\}$-PROP1 allocation.
\end{enumerate}
\end{theorem}

\begin{proof}
For (i), take $\rho=19/(30n)$. We maintain positive slacks and $\Phi^t\le nf(1)$. Both statements hold initially. Fix the realized history and current good, and suppose the statements hold at time $t-1$. Every candidate argument in \eqref{eq:miv-improved-score} is positive, since $s_i^{t-1}+\rho-\rho z_i\ge s_i^{t-1}>0$.

For every agent with $p_i>0$, let $d_i:=f(s_i^{t-1}+\rho-\rho z_i)-f(s_i^{t-1}+\rho)$, and set $d_i:=0$ when $p_i=0$. Assigning the good to $j$ changes the potential by exactly
\begin{equation}
\label{eq:miv-total-change}
    \Phi^t(j)-\Phi^{t-1}=\sum_{i\in N}d_i-\Delta_j.
\end{equation}
Thus it suffices to show that $\sum_i d_i\le\max_i\Delta_i$.

If all $z_i=0$, the potential is unchanged. Otherwise, let $M:=\max_i\Delta_i>0$ and choose $S>\rho$ so that $f(S-\rho)-f(S+1-\rho)=M$.
Such an $S$ exists uniquely because $f$ is strictly decreasing and strictly convex, with limits $+\infty$ at zero and zero at infinity. The comparison proved in Appendix~\ref{app:miv-comparison} gives
\begin{align}
\frac{\sum_i d_i}{M}
&\le n\,\frac{f(S-\rho)-f(S)}{f(S-\rho)-f(S+1-\rho)}
       \min\left\{1,\frac{f(1)}{f(S)}\right\}
\label{eq:miv-total-change-bound}\\
&\le n\rho f(1)<1.\nonumber
\end{align}
The two bounds in the minimum have different roles. One bounds each increase $d_i$ directly and then sums over at most $n$ agents. The other bounds $d_i$ relative to the agent's current potential term and uses $\Phi^{t-1}\le nf(1)$. The latter is stronger for $S\le1$, and the former for $S\ge1$. Together they show that the largest available decrease is enough; no distribution over recipients is required.

By \eqref{eq:miv-total-change}, the maximizing recipient maintains $\Phi^t\le nf(1)$. All updated arguments are positive, and \eqref{eq:miv-function-bounds} now implies $s_i^t>0$, completing the induction. At the final time, $D_i^m=v_i(G)$, so
\[
    p_i+B_i^m>\rho v_i(G)=\frac{19}{30}\frac{v_i(G)}n.
\]
The interpretation following \eqref{eq:miv-accounting} proves (i).

For (ii), take $\rho=1/(n+\kappa)$. Appendix~\ref{app:miv-reciprocal} proves that the reciprocal rule maintains $\Psi^t\le n/(1-\rho)$ and gives $n/(n+\kappa)$-PROP1. For the averaging calculation, each agent with $z_i>0$ receives probability $\rho(1+(1-\rho)z_i/s_i^{t-1})$, which leaves her reciprocal term unchanged in expectation. These probabilities sum to at most $\rho\kappa+\rho(1-\rho)\Psi^{t-1}\le\rho(\kappa+n)=1$.
Choosing between that rule and (i) proves the stated maximum.

All comparisons use only the realized history and the current good, so the claims hold against adaptive adversaries without knowing the horizon. Each rule evaluates one score per agent.
\end{proof}

Combining \Cref{thm:miv-constant} with the one-sided-error transformation of Choo et al.~\cite{choo2026approxproponline} gives the following guarantee for imperfect MIV predictions.

\begin{corollary}
\label{cor:miv-error}
Suppose the algorithm is given $\kappa\in[2,n]$ such that $|N(g)|\le\kappa$ for every good, and a one-sided error bound $\eps\in[0,1)$ for the MIV predictions. Then a deterministic online algorithm guarantees
\[
    \max\left\{
        \frac{19n(1-\eps)}{30n-19\eps},
        \frac{n(1-\eps)}{n+\kappa-\eps}
    \right\}\text{-PROP1}
\]
against adaptive adversaries, without knowing the horizon. Without information about $\kappa$, the first factor applies.
\end{corollary}

\begin{proof}
We give the transformation explicitly. For each agent with $p_i>0$, increase her value for the first good worth at least $(1-\eps)p_i$ to $p_i$, leaving all other values unchanged. Run either rule from \Cref{thm:miv-constant} on these modified values. The one-sided-error condition in \Cref{sec:preliminaries} ensures that this good appears and that no value exceeds $p_i$. Thus the modified instance has exact MIV predictions. The modification is online and, since $\eps<1$, changes no set $N(g)$.

Fix an agent with $p_i>0$. Let $d_i\in[0,\eps p_i]$ be her single value increase and let $h_i:=\max_{g\in G\setminus A_i}v_i(g)$. Suppose the chosen exact-prediction rule gives $n\rho$-PROP1. Returning to the original values reduces the value of a bundle with at most one added good by at most $d_i$, whereas it reduces the total value by $d_i$. Therefore
\[
    v_i(A_i)+h_i\ge\rho v_i(G)-(1-\rho)d_i.
\]
Also, $v_i(A_i)+h_i\ge(1-\eps)p_i$: the changed good is either in $A_i$ or can be added to it. Hence
\[
    d_i\le\frac{\eps}{1-\eps}(v_i(A_i)+h_i),
    \quad
    v_i(A_i)+h_i\ge\frac{\rho(1-\eps)}{1-\rho\eps}\,v_i(G).
\]
The latter inequality follows by substituting the former into the preceding display and rearranging. Agents with $p_i=0$ satisfy PROP1 automatically. Taking $\rho=19/(30n)$ or $\rho=1/(n+\kappa)$ gives the displayed factors. Choose the rule with the larger factor before any good arrives.
\end{proof}

\subsection{The Same PROP1 Factor with Additive Envy Bounds}
\label{subsec:miv-constant-envy}

Assume that the horizon $m$ is known. Define normalized values by
\[
    \widehat v_i(g)
    :=
    \begin{cases}
        v_i(g)/p_i, & p_i>0,\\
        0, & p_i=0.
    \end{cases}
\]
Then $\widehat v_i(g)\in[0,1]$. This normalization does not change the PROP1 factor. The envy bounds below are stated in normalized values; multiplying the bound for agent $i$ by $p_i$ gives the corresponding bound in her original values.
For agents $i,j\in N$, let $E_{ij}^t:=\widehat v_i(A_j^t)-\widehat v_i(A_i^t)$, so $E_{ii}^t=0$.

Retain $\rho=19/(30n)$, the slacks $s_i^t$, and the potential $\Phi^t$ from \Cref{subsec:miv-constant}. For $\eta>0$, define
\begin{equation}
\label{eq:weighted-envy-potential}
    F_\eta^t
    :=
    \sum_{i\ne j}
    e^{\eta E_{ij}^t}
    (1+f(s_j^t+\rho)),
\end{equation}
where a term written as $f(s_j^t+\rho)$ is defined to be zero when $p_j=0$.
The exponential terms follow the approach of Benad\`e et al.~\cite{benade2018envyvanish}. The multiplier links the two fairness objectives: receiving a good may increase other agents' envy toward its recipient, but it also decreases the recipient's PROP1 potential term whenever her adjusted value is positive. Using the \emph{envied agent's} term lets this decrease compensate for the increase in envy.

\paragraph{Allocation rule.}
Set $C:=\frac{789}{500}$, $\eta:=10^{-6}\min\{1,\sqrt{ (n\log n) / m}\}$, $A_\eta:=1+\frac{2\eta^2}{n}$, $F_*:=n(n-1)(1+f(1))$, and $L:= (m\log A_\eta+\log(2CF_*/(C-f(1))))/ \eta$.

The bounds in Appendix~\ref{app:miv-function} give $f(1)<C<30/19$. The first inequality allows a positive initial envy term, while the second ensures that $\Phi^t<nC$ still implies positive slacks. Define
\begin{equation}
\label{eq:constant-envy-potential}
    \Omega^t
    :=
    \frac{\Phi^t}{nC}
    +
    A_\eta^{m-t}e^{-\eta L}F_\eta^t.
\end{equation}
When good $g_t$ arrives, assign it to a recipient that minimizes the resulting value of $\Omega^t$, breaking ties by a fixed order.

\begin{theorem}
\label{thm:miv-constant-envy}
Assume exact MIV predictions and a known horizon $m$. Against adaptive adversaries, the allocation rule above returns an allocation $A$ that is $19/30$-PROP1 and satisfies
\begin{equation}
\label{eq:constant-envy-bound}
    \max_{i,j\in N}
    (\widehat v_i(A_j)-\widehat v_i(A_i))
    \le
    \min\{m,L\},
    \quad
    L=O\!\left(\log n+\sqrt{\frac{m\log n}{n}}\right),
\end{equation}
where the hidden constant is absolute. In particular, if $m\ge n\log n$, the normalized maximum additive envy is $O(\sqrt{m\log n/n})$.
For every known horizon $m$, the bound in \eqref{eq:constant-envy-bound} is also $O(\sqrt{m\log n})$.
Equivalently, $v_i(A_j)-v_i(A_i)\le p_i\min\{m,L\}$ for every $i,j\in N$.
\end{theorem}

\begin{proof}
Initially, $s_i^0+\rho=1$ for every $i$ with $p_i>0$, so $\Phi^0\le nf(1)$ and $F_\eta^0\le F_*$. By the definition of $L$,
\[
    \Omega^0
    \le\frac{f(1)}{C}+A_\eta^m e^{-\eta L}F_*
    =\frac{f(1)}{C}+\frac{C-f(1)}{2C}
    <1.
\]
We prove inductively that $s_i^t>0$ for every $i$ with $p_i>0$ and that $\Omega^t<1$. Suppose these inequalities hold at time $t-1$. Every possible assignment has positive arguments in the potential, since $s_i^{t-1}+\rho-\rho z_i\ge s_i^{t-1}>0$.

Appendix~\ref{app:miv-constant-envy-step} proves that some recipient satisfies $\Omega^t\le\Omega^{t-1}$. The proof compares two possibilities. Either uniform averaging keeps each PROP1 term from increasing and increases $F_\eta$ by a factor of at most $A_\eta$, or one recipient's decrease is enough to reduce the combined potential directly. Thus the allocation rule maintains $\Omega^t<1$. In particular, $\Phi^t<nC<1/\rho$.

If $s_i^t\le0$, its positive argument $s_i^t+\rho$ would be at most $\rho$, and \eqref{eq:miv-function-bounds} would give $f(s_i^t+\rho)>1/\rho$, a contradiction. Hence the selected assignment has positive slacks, completing the induction. The comparison in the appendix fixes only the realized history and current good, so it applies to adaptive adversaries.

At time $m$, positive slacks give us
\[
    p_i+B_i^m>\rho v_i(G)=\frac{19}{30}\frac{v_i(G)}{n}.
\]
The interpretation following \eqref{eq:miv-accounting} proves $19/30$-PROP1. Agents with $p_i=0$ satisfy this condition automatically.
Also, \eqref{eq:constant-envy-potential} and $\Omega^m<1$ imply
\[
    e^{\eta(E_{ij}^m-L)}(1+f(s_j^m+\rho))<1
    \quad(i\ne j).
\]
The second factor is at least one, so $E_{ij}^m<L$. Normalized maximum additive envy is also at most $m$, since all normalized values lie in $[0,1]$.

Finally, $F_*=O(n^2)$, and $C-f(1)$ is a positive absolute constant, so
\[
    L\le\frac{2\eta m}{n}
       +\frac{\log(2CF_*/(C-f(1)))}{\eta}
    =O\!\left(\log n+\sqrt{\frac{m\log n}{n}}\right)
\]
by the choice of $\eta$. If $m\ge n\log n$, the square-root term is at least $\log n$, giving the stated bound. For the $O(\sqrt{m\log n})$ bound, use the bound $m$ when $m\le\log n$; otherwise both terms in the last display are $O(\sqrt{m\log n})$.
\end{proof}

The rule uses $O(n)$ evaluations of $f$ and $O(n^2)$ additional operations per good. Evaluate $f$ at each agent's two possible updated arguments, and compute the needed row and column sums of the pairwise exponential terms once. For a fixed recipient, only her row and column of $(E_{ij}^t)_{i,j\in N}$ change, and only her column uses the receiving value of the multiplier. All candidate values can therefore be compared within the stated bound.

\section{Tight Guarantees for \Like{} and \Rand{}}\label{sec:rand}

Choo et al.~\cite{choo2026approxproponline} showed that \Rand{}, which assigns each good independently and uniformly among all $n$ agents, has a tight $\Theta(1/\log(n/\delta))$ PROP1 guarantee with probability at least $1-\delta$ against a non-adaptive adversary. The classical \Like{} rule of Aleksandrov et al.~\cite{aleksandrov2015onlinefoodbank} instead assigns each good independently and uniformly among the agents who value it positively. Thus, when $i\in N(g)$, agent $i$ receives $g$ with probability $1/|N(g)|$ under \Like{} rather than $1/n$ under \Rand{}. We determine whether this larger probability improves PROP1 when few agents value the same good, and how the two rules compare for additive envy.

As in the result of Choo et al. for \Rand{}, we assume a non-adaptive adversary. Thus the complete valuation sequence is fixed before either rule makes any random choice, as in \Cref{sec:preliminaries}. This is the setting needed for the concentration analysis below. Without future information, the randomized extension of \Cref{thm:online-propk-impossibility} rules out every positive PROP1 approximation against an adaptive adversary, so the non-adaptive assumption is necessary for the guarantees studied here.

Write $A^{\Like}$ and $A^{\Rand}$ for the final allocations. For goods with $N(g)\ne\varnothing$, the rule described above is exactly the classical \Like{} rule of Aleksandrov et al.~\cite{aleksandrov2015onlinefoodbank} under truthful reports. If $N(g)=\varnothing$, \Like{} assigns $g$ arbitrarily; this only makes the allocation complete and does not change any agent's value. Our analysis allows general nonnegative additive valuations.

Recall that $\kappa=\max_{g\in G}|N(g)|$.
If $\kappa\le1$, every good with $N(g)\ne\varnothing$ is assigned by \Like{} to the unique agent in $N(g)$, so \Like{} is proportional. We therefore assume $2\le\kappa\le n$.

We show that \Like{} gains a factor of order $n/\kappa$ in proportionality, and that the dependence on $n$, $\kappa$, and $\delta$ is tight up to constants. The envy comparison goes in the opposite direction because assigning a good among fewer agents increases pairwise variance.

\subsection{Tight High-Probability PROP1 Guarantees}

For a fixed agent, the PROP1 condition is automatic if one good is already worth the target share. Otherwise, every good she values is smaller than that target, while \Like{} assigns each such good to her with probability at least $1/\kappa$. This gives a lower-tail estimate for the value she receives. We state the one-agent estimate separately; the allocation guarantee then follows by a union bound over agents.

\begin{lemma}\label{lem:like-agent-prop1}
Fix an agent $i\in N$ and $\alpha\in(0,1]$, and let $r:=\frac{\alpha\kappa}{n}$.
If $r<1$, then under \Like{},
\[
\Pr[\text{agent $i$ does not satisfy the $\alpha$-PROP1 condition}]
\le
\exp\left(-\frac{3(1-r)^2}{8r}\right).
\]
\end{lemma}

\begin{proof}
If $v_i(G)=0$, the claim is immediate. If $p_i^{\max}\ge\alpha\frac{v_i(G)}{n}$, then agent $i$ satisfies the $\alpha$-PROP1 condition for every realized allocation. Indeed, if a maximum-valued good lies outside $A_i$, add it to $A_i$; if it lies in $A_i$, take the set in the definition of PROP1 to be empty. Hence assume $p_i^{\max}<\alpha\frac{v_i(G)}{n}$.

Set $\mu:=v_i(G)/\kappa$. Then $\alpha v_i(G)/n=r\mu$. For every good $g$ with $v_i(g)>0$, agent $i$ belongs to $N(g)$, so
\[
    \Pr[g\in A_i]=\frac{1}{|N(g)|}\ge\frac1\kappa.
\]
Take independent random variables $U_g\sim\mathrm{Unif}[0,1]$, one for each good with $v_i(g)>0$, and set $I_g:=\mathbf{1} [U_g\le\frac{1}{|N(g)|}]$, $B_g:=\mathbf{1} [U_g\le\frac1\kappa ]$.
Then $I_g\ge B_g$ almost surely. Therefore, we get that
\[
    v_i(A_i)
    =\sum_{g:v_i(g)>0}v_i(g)I_g
    \ge
    Y:=\sum_{g:v_i(g)>0}v_i(g)B_g.
\]
The variables $B_g$ are independent Bernoulli-$1/\kappa$ variables, and $\mathbb E[Y]=\mu$. Every summand of $Y$ is at most $p_i^{\max}<r\mu$, and
\[
    \mathrm{Var}(Y) =\sum_{g:v_i(g)>0}v_i(g)^2\frac1\kappa\left(1-\frac1\kappa\right) \le\frac{p_i^{\max}}{\kappa}\sum_{g:v_i(g)>0}v_i(g)
    <r\mu^2.
\]
The centered variables $v_i(g)(B_g-1/\kappa)$ are independent, have mean zero, and are bounded in absolute value by $r\mu$. If $Y\ge r\mu=\alpha v_i(G)/n$, then agent $i$ satisfies $\alpha$-PROP1 by taking the added set to be empty. Thus Bernstein's inequality gives
\[
\begin{aligned}
\Pr[\text{agent $i$ does not satisfy $\alpha$-PROP1}]
&\le \Pr[Y-\mu<-(1-r)\mu]\\
&\le
\exp\left(
-\frac{(1-r)^2\mu^2}
{2\mathrm{Var}(Y)+\frac{2}{3}(r\mu)(1-r)\mu}
\right) \le
\exp\left(-\frac{3(1-r)^2}{8r}\right),
\end{aligned}
\]
where the last inequality uses $\mathrm{Var}(Y)<r\mu^2$ and $1-r\le1$.
\end{proof}

\begin{theorem}\label{thm:like-prop1}
Fix $\delta\in(0,1)$. On every fixed instance with $2\le\kappa\le n$, \Like{} returns an
$\alpha$-PROP1 allocation with probability at least $1-\delta$, where $\alpha
:= \min \{1,\frac{3n}{32\kappa\log(n/\delta)} \}$.
\end{theorem}

\begin{proof}
Let $r:=\alpha\kappa/n$. Then
\[
r
=
\min\left\{\frac{\kappa}{n},\frac{3}{32\log(n/\delta)}\right\}
\le
\frac{3}{32\log(n/\delta)}
<\frac12.
\]
By \Cref{lem:like-agent-prop1}, each agent does not satisfy the $\alpha$-PROP1 condition with probability at most
\[
\exp\left(-\frac{3(1-r)^2}{8r}\right)
\le
\exp\left(-\frac{3}{32r}\right)
\le
\frac{\delta}{n}.
\]
A union bound over the $n$ agents proves the theorem.
\end{proof}

When every agent values every good, $\kappa=n$, and \Cref{thm:like-prop1} gives the same
$\Theta(1/\log(n/\delta))$ order as \Rand{}. When fewer agents value each good, the factor improves
by $n/\kappa$. In particular, if
$\kappa\le 3n/(32\log(n/\delta))$, then \Like{} returns a PROP1 allocation with probability at least
$1-\delta$.

The next construction shows that both the $n/\kappa$ term and the $\log(n/\delta)$ term are necessary, up to constants, whenever the displayed factor is below one.

\begin{theorem}\label{thm:like-prop1-lower}
There is a universal constant $C>0$ with the following property. Fix $n\ge2$,
$2\le\kappa\le n$, $\delta\in(0,1/2]$, and $\alpha\in(0,1]$. If \Like{} returns an
$\alpha$-PROP1 allocation with probability at least $1-\delta$ on every fixed binary instance
satisfying $\max_{g\in G}|N(g)|\le\kappa$, then $\alpha\le C \cdot \frac{n}{\kappa\log(n/\delta)}$.
\end{theorem}

\begin{proof}
Let $\ell:=\lfloor n/\kappa\rfloor$ and partition the first $\ell\kappa$ agents into
$\ell$ disjoint sets $Q_1,\dots,Q_\ell$, each of size $\kappa$. Set $\mu:=\left\lceil\frac{2n}{\alpha\kappa}\right\rceil$.
For each $h\in[\ell]$, create $\kappa\mu$ goods that have value $1$ for the agents in $Q_h$ and
value $0$ for every other agent. Every good is therefore positively valued by exactly $\kappa$ agents.

Fix an agent $i$ in one of these sets, and let $X_i$ be the number of goods valued by $i$ that
\Like{} allocates to $i$. Then, we have that $X_i\sim\mathrm{Binomial}(\kappa\mu,1/\kappa)$.
Moreover, $v_i(G)=\kappa\mu$ and  $\alpha\frac{v_i(G)}{n}=\frac{\alpha\kappa\mu}{n}\ge2$.
Thus, on the event $F_i:=\{X_i=0\}$, agent $i$ does not satisfy $\alpha$-PROP1: her own value is
zero, and adding one good increases it by at most $1$.

For every such agent,
\[
\Pr[F_i]
=\left(1-\frac1\kappa\right)^{\kappa\mu}
\ge e^{-2\mu}
\ge
\exp\left(-\frac{6n}{\alpha\kappa}\right).
\]
Here we used $\log(1-x)\ge-2x$ for $x\in[0,1/2]$ and
$\mu\le3n/(\alpha\kappa)$; the latter follows from
$2n/(\alpha\kappa)\ge2$.

Let $Z:=\sum_{i=1}^{\ell\kappa}\mathbf{1}_{F_i}$.
If two agents belong to different sets $Q_h$, their events are independent. If they belong to the
same set, then
\[
\Pr[F_i\cap F_j]
=\left(1-\frac2\kappa\right)^{\kappa\mu}
\le
\left(1-\frac1\kappa\right)^{2\kappa\mu}
=\Pr[F_i]\Pr[F_j].
\]
Thus the indicators in the sum defining $Z$ have nonpositive pairwise covariances, and $\mathrm{Var}(Z)\le\mathbb{E}[Z]$.
Since $\ell\kappa\ge n/2$ (because $\lfloor n/\kappa\rfloor\ge n/(2\kappa)$),
\[
\mathbb{E}[Z]
\ge
\frac n2\exp\left(-\frac{6n}{\alpha\kappa}\right).
\]
Since $Z=Z\mathbf 1_{\{Z>0\}}$, Cauchy--Schwarz gives us $\mathbb E[Z]^2 \le \mathbb E[Z^2]\Pr[Z>0]$.
Moreover, $\mathbb E[Z^2]=\mathrm{Var}(Z)+\mathbb E[Z]^2\le\mathbb E[Z]+\mathbb E[Z]^2$. Therefore,
\[
\Pr[Z>0]
\ge
\frac{\mathbb{E}[Z]^2}{\mathbb{E}[Z^2]}
\ge
\frac{\mathbb{E}[Z]}{1+\mathbb{E}[Z]}.
\]
The event $Z>0$ implies that the allocation is not $\alpha$-PROP1. By the assumed guarantee,
$\Pr[Z>0]\le\delta$, and hence $\mathbb{E}[Z]\le\frac{\delta}{1-\delta}\le2\delta$.
Therefore, $\exp\left(-\frac{6n}{\alpha\kappa}\right)
\le
\frac{4\delta}{n}$.
If $\log(n/\delta)\ge2\log4$, then
\[
\frac{6n}{\alpha\kappa}
\ge
\log(n/\delta)-\log4
\ge
\frac12\log(n/\delta),
\]
and hence $\alpha\le12\frac{n}{\kappa\log(n/\delta)}$.
If $\log(n/\delta)<2\log4$, then $\alpha\le1\le2(\log4)n/(\kappa\log(n/\delta))$. Increasing the constant if needed proves
the theorem.
\end{proof}

The same construction shows why \Rand{} cannot benefit from small $\kappa$: its worst-case factor remains of order $1/\log(n/\delta)$ even when every good is valued positively by only one agent.

\begin{theorem}\label{thm:rand-prop1-lower}
There is a universal constant $C>0$ with the following property. Fix $n\ge2$, $\delta\in(0,1/2]$, and $\alpha\in(0,1]$. If \Rand{} returns an $\alpha$-PROP1 allocation with probability at least $1-\delta$ on every fixed binary instance in which $|N(g)|=1$ for every good $g$, then $\alpha\le\frac{C}{\log(n/\delta)}$.
\end{theorem}

\begin{proof}
Set $\mu:=\left\lceil\frac{2}{\alpha}\right\rceil$.
For each agent $i\in N$, create $n\mu$ goods that have value $1$ for agent $i$ and value $0$ for every
other agent. Let $X_i$ be the number of these goods that \Rand{} allocates to agent $i$. Then $X_i\sim\mathrm{Binomial}(n\mu,1/n)$,
and the variables $X_1,\dots,X_n$ are independent because they depend on disjoint sets of goods.
Agent $i$ has total value $n\mu$, so her proportional share is $\mu$. On the event
$F_i:=\{X_i=0\}$, her own value is zero and adding one good increases it by at most $1$, while
$\alpha\mu\ge2$. Thus $F_i$ implies that agent $i$ does not satisfy $\alpha$-PROP1.

Moreover,
\[
\Pr[F_i]
=\left(1-\frac1n\right)^{n\mu}
\ge e^{-2\mu}
\ge\exp\left(-\frac{6}{\alpha}\right),
\]
where we used $\log(1-x)\ge-2x$ for $x\in[0,1/2]$ and
$\mu\le3/\alpha$. Independence gives
\[
\Pr[\Rand{}\text{ does not return an $\alpha$-PROP1 allocation}]
\ge
1-\left(1-e^{-6/\alpha}\right)^n.
\]
By the assumed guarantee, we have that $1-\left(1-e^{-6/\alpha}\right)^n\le\delta$.
Moreover, $(1-e^{-6/\alpha})^n
    \le\exp(-ne^{-6/\alpha})$,
so $\exp(-ne^{-6/\alpha})\ge1-\delta$ and therefore $ne^{-6/\alpha}\le-\log(1-\delta)\le2\delta$.
Consequently,
\[
\frac{6}{\alpha}
\ge
\log\left(\frac{n}{2\delta}\right)
\ge
\frac12\log\left(\frac{n}{\delta}\right),
\]
where the last inequality uses $n/\delta\ge4$. The result follows.
\end{proof}

The preceding results give the exact worst-case orders for the two rules. Combining \Cref{thm:like-prop1,thm:like-prop1-lower} gives the first statement below. The second follows from \Cref{thm:rand-prop1-lower} and the matching guarantee of Choo et al.~\cite{choo2026approxproponline}.

\begin{corollary} \label{cor:random-prop1-summary}
Fix $n\ge2$, $2\le\kappa\le n$, and $\delta\in(0,1/2]$.
\begin{enumerate}[(i)]
    \item On the class of fixed instances satisfying $\max_g|N(g)|\le\kappa$, the largest factor that \Like{} guarantees with probability at least $1-\delta$ is $\Theta\left(
            \min\left\{1,\frac{n}{\kappa\log(n/\delta)}\right\} \right)$.
    \item On all fixed instances, the largest factor that \Rand{} guarantees with probability at least $1-\delta$ is $\Theta\left(\frac{1}{\log(n/\delta)}\right)$.
    The lower bound already holds when every good is positively valued by exactly one agent.
\end{enumerate}
\end{corollary}

\subsection{Why Proportionality and Envy Prefer Different Rules}
\label{subsec:random-envy}

Recall that $\Env(A)$ is the maximum additive envy and that $p_i^{\max}$ is agent $i$'s maximum item value. The next theorem gives bounds for both rules on the same fixed instance. Define
\begin{equation}
\label{eq:random-envy-quantities}
    R_i^{\Like}
    :=
    \sum_{g:v_i(g)>0}\frac{v_i(g)^2}{|N(g)|},
    \quad
    R_i^{\Rand}
    :=
    \frac1n\sum_{g\in G}v_i(g)^2.
\end{equation}

For \Rand{}, the proof follows the pairwise Bernstein approach of Benad\`e et al.~\cite{benade2018envyvanish}, while retaining the instance-specific variance
$\frac{1}{n}\sum_{g\in G}v_i(g)^2$. The novel parts are the corresponding calculation for \Like{}, the direct comparison $R_i^{\Rand}\leq R_i^{\Like}$, and the lower bound showing that the larger quantity for \Like{} can appear in the actual envy.

\begin{theorem}
\label{thm:random-envy-upper}
Fix $\delta\in(0,1)$.
Under \Like{}, with probability at least $1-\delta$,
\[
    \Env(A^{\Like})
    \le
    \max_{i\in N}
    \left(
        2\sqrt{R_i^{\Like} \log\frac{n(n-1)}{\delta}}
        +\frac{p_i^{\max} \log\frac{n(n-1)}{\delta}}{3}
    \right).
\]
Under \Rand{}, with probability at least $1-\delta$,
\[
    \Env(A^{\Rand})
    \le
    \max_{i\in N}
    \left(
        2\sqrt{R_i^{\Rand} \log\frac{n(n-1)}{\delta}}
        +\frac{p_i^{\max} \log\frac{n(n-1)}{\delta}}{3}
    \right).
\]
Moreover, $R_i^{\Rand}\le R_i^{\Like}$ for every agent $i \in N$.
\end{theorem}

\begin{proof}
Fix distinct agents $i,j$, and for each good $g$ define $X_g := v_i(g)(\mathbf 1[g\in A_j]-\mathbf 1[g\in A_i])$.
Then
\[
    v_i(A_j)-v_i(A_i)=\sum_{g\in G}X_g.
\]
The variables $X_g$ are independent across goods under both rules.

Consider \Like{} first. If $v_i(g)=0$, then $X_g=0$. Suppose $v_i(g)>0$. If $j\in N(g)$, then $X_g$ equals $v_i(g)$ with probability $1/|N(g)|$, equals $-v_i(g)$ with probability $1/|N(g)|$, and is zero otherwise. If $j\notin N(g)$, then $X_g$ equals $-v_i(g)$ with probability $1/|N(g)|$ and is zero otherwise. In either case,
\[
    \mathbb E[X_g]\le0,
    \quad
    \mathrm{Var}(X_g)\le\frac{2v_i(g)^2}{|N(g)|},  \text{ and} \quad 
    |X_g-\mathbb E[X_g]|\le v_i(g)\le p_i^{\max}.
\]
Therefore, $\sum_{g\in G}\mathrm{Var}(X_g)
    \le
    2R_i^{\Like}$.
One-sided Bernstein's inequality gives us
\[
\Pr\left[
    v_i(A_j)-v_i(A_i)
    >
    2\sqrt{R_i^{\Like}x}+\frac{p_i^{\max} x}{3}
\right]
\le e^{-x}.
\]
Here we used $\mathbb E[\sum_g X_g]\le0$.

Under \Rand{}, each good goes to $i$ with probability $1/n$ and to $j$ with probability $1/n$. Hence $\mathbb E[X_g]=0$, $\mathrm{Var}(X_g)=\frac{2v_i(g)^2}{n}$, $|X_g|\le p_i^{\max}$, and the same inequality holds with $R_i^{\Rand}$ in place of $R_i^{\Like}$.

There are $n(n-1)$ ordered pairs of distinct agents. Since $e^{-x}=\delta/(n(n-1))$, a union bound proves the corresponding inequality for each rule.

Finally, whenever $v_i(g)>0$, we have $|N(g)|\le n$, and therefore $\frac{v_i(g)^2}{n} \le    \frac{v_i(g)^2}{|N(g)|}$.
Summing over the goods proves $R_i^{\Rand}\le R_i^{\Like}$.
\end{proof}

\begin{corollary}
\label{cor:random-envy-bounded}
If $v_i(g)\in[0,1]$ for every agent and good, then under \Rand{}, with probability at least $1-\delta$,
\[
    \Env(A)
    \le
    2\sqrt{\frac{m}{n}\log\frac{n(n-1)}{\delta}}
    +
    \frac13\log\frac{n(n-1)}{\delta}.
\]
\end{corollary}

\begin{proof}
For every agent $i$, $R_i^{\Rand}\le m/n$ and $p_i^{\max}\le1$. Apply \Cref{thm:random-envy-upper}.
\end{proof}

\begin{corollary}
\label{cor:random-envy-expectation}
For each rule $\mathcal R\in\{\Like,\Rand\}$,
\[
    \mathbb E[\Env(A^{\mathcal R})]
    =
    O\left(
        \max_i\left\{
            \sqrt{R_i^{\mathcal R}\log n}+p_i^{\max}\log n
        \right\}
    \right).
\]
\end{corollary}

\begin{proof}
Fix one of the two rules, and set $R:=\max_i R_i^{\mathcal R}$, $L:=\log(n(n-1))$.
For every $u>0$, apply \Cref{thm:random-envy-upper} with $\delta=e^{-u}$. Using $\sqrt{L+u}\le\sqrt L+\sqrt u$ gives
\[
\Pr\left[
    \Env(A^{\mathcal R})
    >
    2\sqrt{RL}+\frac{L}{3}\max_i p_i^{\max}
    +2\sqrt{Ru}+\frac{u}{3}\max_i p_i^{\max}
\right]
\le e^{-u}.
\]
Let $Z$ be an exponential random variable with mean one, so $\Pr[Z>u]=e^{-u}$. Since $u \mapsto 2\sqrt{Ru}+\frac{u}{3}\max_i p_i^{\max}$ is increasing, the preceding inequality and the tail-integration formula give
\[
    \mathbb E[\Env(A^{\mathcal R})]
    \le{} 2\sqrt{RL}+\frac{L}{3}\max_i p_i^{\max}+\mathbb E\left[2\sqrt{RZ}+\frac{Z}{3}\max_i p_i^{\max}\right].
\]
Using $\mathbb E[\sqrt Z]=\sqrt\pi/2$ and $\mathbb E[Z]=1$ gives us
\[
    \mathbb E[\Env(A^{\mathcal R})]
    \le
    2\sqrt{RL}
    +\frac{L}{3}\max_i p_i^{\max}
    +\sqrt{\pi R}
    +\frac13\max_i p_i^{\max}.
\]
Since $L=\Theta(\log n)$ for $n\ge2$, the stated bound follows.
\end{proof}

The quantities in \eqref{eq:random-envy-quantities} explain the difference between the two rules. Under \Rand{}, a good contributes $v_i(g)^2/n$ to agent $i$'s variance bound. Under \Like{}, the contribution is $v_i(g)^2/|N(g)|$, which is larger when few agents value that good positively. The next result shows that this difference can appear in the actual envy, even for binary valuations.

\begin{proposition}
\label{prop:like-envy-lower}
For every $2\le\kappa\le n$ and every integer $m\ge\kappa$, there is a binary instance with $\max_{g\in G}|N(g)|=\kappa$ such that, under \Like{}, $\mathbb E[\Env(A)]
    \ge \frac{9}{64\sqrt2} \cdot \sqrt{\frac{m}{\kappa}}$.
\end{proposition}

\begin{proof}
Let $Q:=\{1,\dots,\kappa\}$. Every agent in $Q$ values every good at $1$, and every agent outside $Q$ values every good at $0$. Thus $N(g)=Q$ for every good, and \Like{} assigns each good independently and uniformly among the agents in $Q$.

For $t=1,\dots,m$, define
\[
Z_t:=
\begin{cases}
    1, & g_t\text{ is assigned to agent }2,\\
    -1, & g_t\text{ is assigned to agent }1,\\
    0, & \text{otherwise},
\end{cases}
\quad
D:=\sum_{t=1}^m Z_t.
\]
Then $v_1(A_2)-v_1(A_1)=D$, $v_2(A_1)-v_2(A_2)=-D$, so $\Env(A)\ge|D|$. The variables $Z_t$ are independent and satisfy $\mathbb E[Z_t]=0$, $\mathbb E[Z_t^2]=\mathbb E[Z_t^4]=\frac{2}{\kappa}$.
Set $s:=\sqrt{2m/\kappa}$. Then $\mathbb E[D^2]=s^2$. Since $m\ge\kappa$,
\begin{align*}
    \mathbb E[D^4] =
    \sum_{t=1}^m\mathbb E[Z_t^4]
    + 6\sum_{1\le r<t\le m}\mathbb E[Z_r^2]\mathbb E[Z_t^2] \le
    s^2+3s^4 \le 4s^4.
\end{align*}
Applying the Paley--Zygmund inequality to $D^2$ gives us
\[
    \Pr\left[|D|\ge\frac{s}{2}\right]
    \ge
    \left(1-\frac14\right)^2
    \frac{\mathbb E[D^2]^2}{\mathbb E[D^4]}
    \ge
    \frac{9}{64}.
\]
Therefore
\[
    \mathbb E[\Env(A)]
    \ge
    \mathbb E[|D|]
    \ge
    \frac{9}{64}\frac{s}{2}
    =
    \frac{9}{64\sqrt2}\sqrt{\frac{m}{\kappa}}. \qedhere
\]
\end{proof}

On the instance in \Cref{prop:like-envy-lower}, the variance quantities are $R_i^{\Like}=m/\kappa$ and $R_i^{\Rand}=m/n$ for the agents in $Q$. Thus the lower bound has the same square-root dependence on $m/\kappa$ as the upper bound in \Cref{thm:random-envy-upper}, apart from the logarithmic factor needed for a simultaneous high-probability guarantee over all pairs. Together with \Cref{thm:like-prop1,thm:rand-prop1-lower}, this gives a clear comparison between the two rules: \Like{} can give much better proportionality when few agents value each good, while \Rand{} has the smaller variance bound for additive envy.

The comparison above is only between \Like{} and \Rand{}; the optimal envy guarantee for general online algorithms is a separate question. Halpern et al.~\cite{halpern2025onlineenvy} study that question through multicolor discrepancy.

\section{Conclusion}
\label{sec:conclusion}

This paper closes three gaps in online fair division. 
First, without predictions, every positive PROP$k$ factor is impossible for goods, while every PROP$k$ factor below $n$ is impossible for chores; both statements apply to randomized online algorithms against adaptive adversaries. Second, maximum item value predictions remove the previous dependence on $n$, improving $1/n$-PROP1 to $19/30$-PROP1 and to $\max\{19/30,n/(n+\kappa)\}$ when the algorithm knows an upper bound $\kappa$ on the number of agents who value any good positively. With exact predictions and a known horizon $m\ge n\log n$, we retain $19/30$-PROP1 together with $O(\sqrt{m\log n/n})$ normalized maximum additive envy, with an absolute hidden constant. Third, against a non-adaptive adversary, the classical \Like{} rule has a tight $\Theta(\min\{1,n/(\kappa\log(n/\delta))\})$ guarantee, while \Rand{} worst-case factor remains
$\Theta(1/\log(n/\delta))$ even when each good is valued by only one agent.

The results also give a common explanation for when online proportionality becomes possible. A maximum item value prediction gives the value of a maximum-valued good, which can be used to account for the one good allowed by PROP1 without revealing which good attains that value or when it arrives. A non-adaptive adversary instead makes concentration available. Without either source of predictability, an adaptive adversary rules out every positive PROP$k$ factor for goods and every PROP$k$ factor below $n$ for chores. The parameter $\kappa$ further shows that proportionality and envy can favor different random rules: fewer positive-valuing agents help \Like{} deliver value, but increase the pairwise fluctuations relevant to additive envy.

Several questions remain. What is the optimal PROP1 factor with exact MIV predictions? Against a non-adaptive adversary, can a general randomized online algorithm improve on \Rand{}'s $\Theta(1/\log(n/\delta))$ worst-case PROP1 factor on unrestricted instances?

\subsection*{Declaration of generative AI use}
The key ideas, arguments, and results presented in an initial version of this work were developed by the authors. GPT-5.6 Sol was used to assist with refining the exposition, including improving clarity, phrasing, and presentation. GPT-6 was then used to improve the initial $1/2$-PROP1 guarantee to $19/30$-PROP1, the authors verified its correctness and improved the argument. It also assisted in extending the simultaneous PROP1 and additive-envy guarantee to the same factor.
The authors take full responsibility for the content and conclusions of the work.

\bibliographystyle{alpha}
\bibliography{bib}

\newcommand{\etalchar}[1]{$^{#1}$}
\begin{thebibliography}{BGGS21b}

\bibitem[AAB{\etalchar{+}}23]{amanatidis2023fairdivisionsurvey}
Georgios Amanatidis, Haris Aziz, Georgios Birmpas, Aris Filos-Ratsikas, Bo~Li, Herv{\'e} Moulin, Alexandros~A. Voudouris, and Xiaowei Wu.
\newblock Fair division of indivisible goods: Recent progress and open questions.
\newblock {\em Artificial Intelligence}, 322:103965, 2023.

\bibitem[AAGW15]{aleksandrov2015onlinefoodbank}
Martin Aleksandrov, Haris Aziz, Serge Gaspers, and Toby Walsh.
\newblock Online fair division: Analysing a food bank problem.
\newblock In {\em Proceedings of the 24th International Joint Conference on Artificial Intelligence (IJCAI)}, pages 2540--2546, 2015.

\bibitem[ABM18]{amanatidis2018ef}
Georgios Amanatidis, Georgios Birmpas, and Vangelis Markakis.
\newblock Comparing approximate relaxations of envy-freeness.
\newblock In {\em Proceedings of the 27th International Joint Conference on Artificial Intelligence (IJCAI)}, pages 42--48, 2018.

\bibitem[AGMP26]{amanatidis2026buffers}
Georgios Amanatidis, Giulio Giaconi, Evangelos Markakis, and Nicos Protopapas.
\newblock Online fair division meets reordering buffers.
\newblock {\em arXiv preprint arXiv:2607.01159}, 2026.

\bibitem[ALMT25]{amanatidis2025personalized2value}
Georgios Amanatidis, Alexandros Lolos, Evangelos Markakis, and Victor Turmel.
\newblock Online fair division for personalized 2-value instances.
\newblock In {\em Proceedings of the 18th International Symposium on Algorithmic Game Theory (SAGT)}, pages 209--227, 2025.

\bibitem[AMS20]{aziz2020prop1}
Haris Aziz, Herv{\'e} Moulin, and Fedor Sandomirskiy.
\newblock A polynomial-time algorithm for computing a pareto optimal and almost proportional allocation.
\newblock {\em Operations Research Letters}, 48(5):573--578, 2020.

\bibitem[ASH26]{adams2026perpetuallyfairassignmentsbalanced}
Terrence Adams and Erel Segal-Halevi.
\newblock Perpetually fair assignments via balanced sequences of permutations.
\newblock {\em arXiv preprint arXiv:2602.21687}, 2026.

\bibitem[AW20]{aleksandrov2020online}
Martin Aleksandrov and Toby Walsh.
\newblock Online fair division: A survey.
\newblock In {\em Proceedings of the 34th AAAI Conference on Artificial Intelligence (AAAI)}, pages 13557--13562, 2020.

\bibitem[BBKN18]{BarmanBiswasKrishnamurthyNarahari2018}
Siddharth Barman, Arpita Biswas, Sanath~Kumar Krishnamurthy, and Yadati Narahari.
\newblock Groupwise maximin fair allocation of indivisible goods.
\newblock In {\em Proceedings of the 32nd AAAI Conference on Artificial Intelligence (AAAI)}, pages 917--924, 2018.

\bibitem[BEF23]{babaioff2023fairshare}
Moshe Babaioff, Tomer Ezra, and Uriel Feige.
\newblock Fair-share allocations for agents with arbitrary entitlements.
\newblock {\em Mathematics of Operations Research}, 49(4):2180--2211, 2023.

\bibitem[BGGS21a]{BaklanovGarimidiGkatzelisSchoepflin2021AAAI}
Artem Baklanov, Pranav Garimidi, Vasilis Gkatzelis, and Daniel Schoepflin.
\newblock Achieving proportionality up to the maximin item with indivisible goods.
\newblock In {\em Proceedings of the 35th AAAI Conference on Artificial Intelligence (AAAI)}, pages 5143--5150, 2021.

\bibitem[BGGS21b]{BaklanovGarimidiGkatzelisSchoepflin2021IJCAI}
Artem Baklanov, Pranav Garimidi, Vasilis Gkatzelis, and Daniel Schoepflin.
\newblock {PROPm} allocations of indivisible goods to multiple agents.
\newblock In {\em Proceedings of the 30th International Joint Conference on Artificial Intelligence (IJCAI)}, pages 24--30, 2021.

\bibitem[BKPP18]{benade2018envyvanish}
Gerdus Benad{\`e}, Aleksandr~M. Kazachkov, Ariel~D. Procaccia, and Christos-Alexandros Psomas.
\newblock How to make envy vanish over time.
\newblock In {\em Proceedings of the 19th ACM Conference on Economics and Computation (EC)}, pages 593--610, 2018.

\bibitem[BT96]{brams1996fairdivision}
Steven~J. Brams and Alan~D. Taylor.
\newblock {\em Fair Division: From Cake-Cutting to Dispute Resolution}.
\newblock Cambridge University Press, 1996.

\bibitem[Bud11]{Budish2011}
Eric Budish.
\newblock The combinatorial assignment problem: Approximate competitive equilibrium from equal incomes.
\newblock {\em Journal of Political Economy}, 119(6):1061--1103, 2011.

\bibitem[CFK{\etalchar{+}}26]{choo2026approxproponline}
Davin Choo, Winston Fu, Derek Khu, Tzeh~Yuan Neoh, Tze-Yang Poon, and Nicholas Teh.
\newblock Approximate proportionality in online fair division.
\newblock In {\em Proceedings of the 43rd International Conference on Machine Learning (ICML)}, 2026.
\newblock Extended version available as arXiv:2508.03253.

\bibitem[CFS17]{conitzer2017fairpublic}
Vincent Conitzer, Rupert Freeman, and Nisarg Shah.
\newblock Fair public decision making.
\newblock In {\em Proceedings of the 18th ACM Conference on Economics and Computation (EC)}, pages 629--646, 2017.

\bibitem[CGR{\etalchar{+}}23]{CaragiannisGargRathiSharmaVarricchio2023}
Ioannis Caragiannis, Jugal Garg, Nidhi Rathi, Eklavya Sharma, and Giovanna Varricchio.
\newblock New fairness concepts for allocating indivisible items.
\newblock In {\em Proceedings of the 32nd International Joint Conference on Artificial Intelligence (IJCAI)}, pages 2554--2562, 2023.

\bibitem[CKM{\etalchar{+}}19]{caragiannis2019unreasonable}
Ioannis Caragiannis, David Kurokawa, Herv{\'e} Moulin, Ariel~D. Procaccia, Nisarg Shah, and Junxing Wang.
\newblock The unreasonable fairness of maximum {N}ash welfare.
\newblock {\em ACM Transactions on Economics and Computation}, 7(3):12:1--12:32, 2019.

\bibitem[CLT26]{choi2026tfdmm}
Kui-Wang Choi, Minming Li, and Nicholas Teh.
\newblock Temporal fair division of indivisible mixed manna: Tractable settings.
\newblock {\em arXiv preprint}, arXiv:2608.20033, 2026.

\bibitem[CT26]{chen2026competitiveanalysis}
Tianqi Chen and Zhiyi Tan.
\newblock Competitive analysis for online fair division under multiple fairness notions.
\newblock {\em arXiv preprint arXiv:2606.15404}, 2026.

\bibitem[CTGW26]{cohen2026budgetconstraints}
Saar Cohen, Nicholas Teh, Paul~W. Goldberg, and Michael~J. Wooldridge.
\newblock Online fair division with budget constraints.
\newblock {\em arXiv preprint arXiv:2607.23310}, 2026.

\bibitem[ELL{\etalchar{+}}25]{elkind2025temporal}
Edith Elkind, Alexander Lam, Mohamad Latifian, Tzeh~Yuan Neoh, and Nicholas Teh.
\newblock Temporal fair division of indivisible items.
\newblock In {\em Proceedings of the 24th International Conference on Autonomous Agents and Multiagent Systems (AAMAS)}, pages 676--685, 2025.

\bibitem[GRT26]{Goldberg2026}
Paul Goldberg, Isaac Robinson, and Nicholas Teh.
\newblock Minimizing cumulative envy in allocating a sequence of items.
\newblock In {\em Proceedings of the 19th International Symposium on Algorithmic Game Theory (SAGT)}, 2026.
\newblock Forthcoming.

\bibitem[GS26]{garg2026relations}
Jugal Garg and Eklavya Sharma.
\newblock Exploring relations among fairness notions in discrete fair division.
\newblock In {\em Proceedings of the 25th International Conference on Autonomous Agents and Multiagent Systems (AAMAS)}, pages 2151--2159, 2026.

\bibitem[HPVX25]{halpern2025onlineenvy}
Daniel Halpern, Alexandros Psomas, Paritosh Verma, and Daniel Xie.
\newblock Online envy minimization and multicolor discrepancy: Equivalences and separations.
\newblock In {\em Proceedings of the 26th ACM Conference on Economics and Computation (EC)}, page 188, 2025.

\bibitem[KM25]{KobayashiMahara2025}
Yusuke Kobayashi and Ryoga Mahara.
\newblock Proportional allocation of indivisible goods up to the least valued good on average.
\newblock {\em SIAM Journal on Discrete Mathematics}, 39(1):533--549, 2025.

\bibitem[KMS25]{KulkarniMehtaShahkar2025}
Pooja Kulkarni, Ruta Mehta, and Parnian Shahkar.
\newblock Online fair division: Towards ex-post constant {MMS} guarantees.
\newblock In {\em Proceedings of the 26th ACM Conference on Economics and Computation (EC)}, page 638, 2025.

\bibitem[KSHH26]{kahana2026perpetual}
Ido Kahana, Erel Segal-Halevi, and Noam Hazon.
\newblock Perpetual fully-online approximate fairness.
\newblock {\em arXiv preprint arXiv:2605.19844}, 2026.

\bibitem[LMMS04]{lipton2004ece}
Richard~J. Lipton, Evangelos Markakis, Elchanan Mossel, and Amin Saberi.
\newblock On approximately fair allocations of indivisible goods.
\newblock In {\em Proceedings of the 5th ACM Conference on Electronic Commerce (EC)}, pages 125--131, 2004.

\bibitem[Mou03]{moulin2003fairdivision}
Herv{\'e} Moulin.
\newblock {\em Fair Division and Collective Welfare}.
\newblock MIT Press, 2003.

\bibitem[MP26]{melissourgos2026onlineefx}
Themistoklis Melissourgos and Nicos Protopapas.
\newblock Online {EFX} allocations with predictions.
\newblock In {\em Proceedings of the 25th International Conference on Autonomous Agents and Multiagent Systems (AAMAS)}, pages 3519--3521, 2026.

\bibitem[NPT26]{neoh2026online}
Tzeh~Yuan Neoh, Jannik Peters, and Nicholas Teh.
\newblock Online fair division with additional information.
\newblock In {\em Proceedings of the 43rd International Conference on Machine Learning (ICML)}, 2026.
\newblock Extended version available as arXiv:2505.24503.

\bibitem[SS25]{seddighin2025lowerbound}
Masoud Seddighin and Saeed Seddighin.
\newblock Lower bound for online {MMS} assignment of indivisible chores.
\newblock {\em arXiv preprint arXiv:2507.12984}, 2025.

\bibitem[STWZ25]{song2025onlinemmschores}
Jiaxin Song, Biaoshuai Tao, Wenqian Wang, and Yuhao Zhang.
\newblock Online {MMS} allocation for chores.
\newblock {\em arXiv preprint arXiv:2507.14039}, 2025.

\bibitem[WW26]{wang2026onlinefairBinary}
Yuanyuan Wang and Tianze Wei.
\newblock Online fair allocations with binary valuations and beyond.
\newblock In {\em Proceedings of the 40th AAAI Conference on Artificial Intelligence (AAAI)}, pages 17267--17275, 2026.

\bibitem[ZBW23]{zhou2023icml_mms_chores}
Shengwei Zhou, Rufan Bai, and Xiaowei Wu.
\newblock Multi-agent online scheduling: {MMS} allocations for indivisible items.
\newblock In {\em Proceedings of the 40th International Conference on Machine Learning (ICML)}, pages 42506--42516, 2023.

\end{thebibliography}

\newpage 

\appendix

\section{Definitions and Consequences for Other Fairness Notions}
\label{app:fairness-implications}

We use the notation of \Cref{sec:preliminaries}. For each agent $i$, let
$h_i:=\max_{g\in G\setminus A_i}v_i(g)$, where $h_i=0$ when $A_i=G$.
By the equivalent form of $\alpha$-PROP1 in
\Cref{sec:preliminaries}, for every $\alpha\in[0,1]$, $A$ is
$\alpha$-PROP1 exactly when
\begin{equation}
\label{eq:prop1-max-good-form}
    v_i(A_i)+h_i
    \ge
    \alpha\frac{v_i(G)}{n}
    \quad\text{for every agent }i\in N.
\end{equation}

\subsection{Definitions}

We use the direct multiplicative versions of the following notions. 

\paragraph{Proportionality-based axioms.}
An allocation is $\alpha$-PROP if $v_i(A_i)\ge\alpha \cdot \frac{v_i(G)}{n}$ for every agent $i \in N$. It is $\alpha$-PROPX if, for every $i$ and every good $g\in G\setminus A_i$, $v_i(A_i)+v_i(g)\ge\alpha\frac{v_i(G)}{n}$; when $A_i=G$, the condition for $i$ is automatic \cite{aziz2020prop1}.

For distinct agents $i,j$, let
\[
    \ell_{ij}:=
    \begin{cases}
        \min_{g\in A_j}v_i(g), & A_j\ne\varnothing,\\
        0, & A_j=\varnothing.
    \end{cases}
\]
An allocation is $\alpha$-PROPm, $\alpha$-PROPavg, or $\alpha$-Avg-EFX if, respectively,
\begin{align*}
    v_i(A_i)+\max_{j\ne i}\ell_{ij} \ge\alpha\frac{v_i(G)}{n},\quad 
    v_i(A_i)+\frac{1}{n-1}\sum_{j\ne i}\ell_{ij} \ge\alpha\frac{v_i(G)}{n}, \quad
    v_i(A_i)+\frac{1}{n}\sum_{j\ne i}\ell_{ij}
    &\ge\alpha\frac{v_i(G)}{n}
\end{align*}
for every agent $i$ \cite{BaklanovGarimidiGkatzelisSchoepflin2021AAAI,BaklanovGarimidiGkatzelisSchoepflin2021IJCAI,KobayashiMahara2025}. At $\alpha=1$, these are the usual exact notions.

\paragraph{Envy-based and share-based axioms.}
An allocation is $\alpha$-EF if $v_i(A_i)\ge\alpha v_i(A_j)$ for every $i,j$. It is $\alpha$-EFX if, for every $i,j$ and every $g\in A_j$ with $v_i(g)>0$, $v_i(A_i)\ge\alpha v_i(A_j\setminus\{g\})$.
It is $\alpha$-EF1 if, for every $i,j$, there is a set $R_{ij}\subseteq A_j$ with $|R_{ij}|\le1$ such that $v_i(A_i)\ge\alpha v_i(A_j\setminus R_{ij})$ \cite{lipton2004ece,amanatidis2018ef}.

For a positive integer $q$ and a set of goods $S$, let $\Pi_q(S)$ be the set of partitions of $S$ into $q$ bundles. Agent $i$'s maximin share is $\mms_i:=\max_{(X_1,\dots,X_n)\in\Pi_n(G)}\min_{j\in N}v_i(X_j)$.
An allocation is $\alpha$-MMS if $v_i(A_i)\ge\alpha\mms_i$ for every agent $i$ \cite{Budish2011}.

An allocation is $\alpha$-EEFX if, for every agent $i$, there is an allocation $B^i=(B_1^i,\dots,B_n^i)$ of $G$ such that $B_i^i=A_i$ and agent $i$ satisfies $\alpha$-EFX in $B^i$. Let $\mathcal E_i$ be the set of allocations in which agent $i$ satisfies EFX, and define $\mathrm{MXS}_i:=\min_{B\in\mathcal E_i}v_i(B_i)$.
The set $\mathcal E_i$ is nonempty: the allocation with $B_i=G$ and
$B_j=\varnothing$ for every $j\ne i$ belongs to $\mathcal E_i$.
Since $G$ is finite, the minimum is attained.
An allocation is $\alpha$-MXS if $v_i(A_i)\ge\alpha\,\mathrm{MXS}_i$ for every agent $i$ \cite{CaragiannisGargRathiSharmaVarricchio2023}.

For a nonempty set of agents $S\subseteq N$ and $i\in S$, let $\mu_i(S) := \max_{(X_1,\dots,X_{|S|})\in\Pi_{|S|}(\bigcup_{j\in S}A_j)}
    \min_{r\in[|S|]}v_i(X_r)$.
An allocation is $\alpha$-GMMS if $v_i(A_i)\ge\alpha\mu_i(S)$ for every nonempty $S\subseteq N$ and every $i\in S$ \cite{BarmanBiswasKrishnamurthyNarahari2018}.

\subsection{Implications}

\begin{lemma}
\label{lem:prop1-implications}
Fix $\alpha\in(0,1]$. Each of the following conditions implies
$\alpha$-PROP1:
\[
\begin{gathered}
\alpha\text{-PROP},\ \alpha\text{-PROP1},\ \alpha\text{-PROPX},\
\alpha\text{-PROPm},\ \alpha\text{-PROPavg},\ \alpha\text{-Avg-EFX},\\
\alpha\text{-EF},\ \alpha\text{-EFX},\ \alpha\text{-EF1},\
\alpha\text{-EEFX},\ \alpha\text{-MXS},\ \alpha\text{-MMS},\
\alpha\text{-GMMS}.
\end{gathered}
\]
\end{lemma}

\begin{proof}
The implications from $\alpha$-PROP and $\alpha$-PROP1 are immediate. For $\alpha$-PROPX, if $A_i\ne G$, apply its defining inequality to a good attaining $h_i$; if $A_i=G$, the PROP1 condition is automatic. For PROPm, PROPavg, and Avg-EFX, every $\ell_{ij}$ is at most $h_i$, so each correction term in the definition is at most $h_i$. In every case, \eqref{eq:prop1-max-good-form} follows.

Next, $\alpha$-EF implies $\alpha$-EFX, and $\alpha$-EFX implies
$\alpha$-EF1. To prove the remaining implication, suppose that $A$ is
$\alpha$-EF1 and fix an agent $i$. If $A_i=G$, the PROP1 condition is
immediate. For each $j\ne i$, the $\alpha$-EF1 condition allows the removal of at most one good from $A_j$. Since every good in $A_j$ has value at most $h_i$ to agent $i$, $v_i(A_j)\le \frac{v_i(A_i)}{\alpha}+h_i$.
Therefore
\[
    v_i(G)
    \le
    v_i(A_i)+(n-1)\left(\frac{v_i(A_i)}{\alpha}+h_i\right).
\]
Rearranging gives
\[
    v_i(A_i)+h_i-\alpha\frac{v_i(G)}{n}
    \ge
    \frac{1-\alpha}{n}v_i(A_i)
    +
    \frac{n-\alpha(n-1)}{n}h_i
    \ge0.
\]
Thus agent $i$ satisfies $\alpha$-PROP1.

The same calculation applies to $\alpha$-EEFX. Fix $i$ and choose an allocation $B^i$ with $B_i^i=A_i$ in which agent $i$ satisfies $\alpha$-EFX. For $j\ne i$, if $v_i(B_j^i)=0$, then
$v_i(B_j^i)\le v_i(A_i)/\alpha+h_i$ is immediate. Otherwise, choose a good $g_j\in B_j^i$ of maximum value to agent $i$. Then $v_i(g_j)>0$, and $v_i(g_j)\le h_i$ because $B_j^i\subseteq G\setminus A_i$. Moreover, $\alpha$-EFX gives
\[
    v_i(A_i)\ge \alpha v_i(B_j^i\setminus\{g_j\}).
\]
Hence $v_i(B_j^i)\le v_i(A_i)/\alpha+h_i$. Summing over $j\ne i$ and applying the same rearrangement as above gives $v_i(A_i)+h_i\ge\alpha v_i(G)/n$.

Now suppose that $A$ is $\alpha$-MXS. Fix an agent $i$ with $A_i\ne G$ and assume for contradiction that
\begin{equation}
\label{eq:mxs-contradiction}
    v_i(A_i)+h_i<\alpha\frac{v_i(G)}{n}.
\end{equation}
Choose $Y\in\mathcal E_i$ with $v_i(Y_i)=\mathrm{MXS}_i$. Since $v_i(A_i)\ge\alpha\,\mathrm{MXS}_i$, we must have $\mathrm{MXS}_i<v_i(G)/n$; otherwise $v_i(A_i)\ge\alpha v_i(G)/n$. The bundles in $Y$ have total value $v_i(G)$ to agent $i$, so some $k\ne i$ satisfies $v_i(Y_k)>v_i(G)/n$. Since
\[
    v_i(A_i)<\alpha\frac{v_i(G)}{n}
    \le\frac{v_i(G)}{n}
    <v_i(Y_k),
\]
additivity and nonnegativity give
\[
    v_i(Y_k\setminus A_i)
    =v_i(Y_k)-v_i(Y_k\cap A_i)
    \ge v_i(Y_k)-v_i(A_i)
    >0.
\]
Hence there is a good $g\in Y_k\setminus A_i$ with $v_i(g)>0$.
Agent $i$ satisfies EFX in $Y$, and hence
\[
    \mathrm{MXS}_i=v_i(Y_i)
    \ge v_i(Y_k\setminus\{g\})
    >\frac{v_i(G)}{n}-v_i(g).
\]
Using $v_i(g)\le h_i$ and $\alpha\le1$,
\[
    v_i(A_i)+h_i
    \ge \alpha\,\mathrm{MXS}_i+h_i
    >\alpha\left(\frac{v_i(G)}{n}-v_i(g)\right)+h_i
    \ge\alpha\frac{v_i(G)}{n},
\]
contradicting \eqref{eq:mxs-contradiction}.

It remains to consider MMS and GMMS. The exact implication from MMS to PROP1 is known \cite{CaragiannisGargRathiSharmaVarricchio2023}; we give the multiplicative calculation. Fix an agent $i$ with $A_i\ne G$ and assume for contradiction that $v_i(A_i)+h_i<\alpha\frac{v_i(G)}{n}$.
Set $\tau:=v_i(A_i)/\alpha$. Since $A$ is $\alpha$-MMS,
$\mms_i\le\tau$. The assumed inequality gives
\[
    \tau<\frac{v_i(G)}{n},
    \quad
    h_i
    <
    \alpha\left(\frac{v_i(G)}{n}-\tau\right)
    \le
    \frac{v_i(G)}{n}-\tau.
\]
We construct an $n$-partition in which every bundle has value more than $\tau$ to agent $i$. Start with $A_i$ and add goods from $G\setminus A_i$ until its value first exceeds $\tau$. Every added good has value at most $h_i$, so the completed bundle has value less than $v_i(G)/n$. Form each later bundle from the remaining goods, stopping when its value first exceeds $\tau$. Again, each completed bundle has value less than $v_i(G)/n$. Before the $r$th bundle is formed, the first $r-1$ bundles have total value less than $(r-1)v_i(G)/n$, so the remaining goods have value more than
\[
    v_i(G)-\frac{(r-1)v_i(G)}{n}
    \ge\frac{v_i(G)}{n}
    >\tau.
\]
Thus all $n$ bundles can be formed. Add any goods left afterward to arbitrary bundles. The least valuable bundle is worth more than $\tau$, contradicting $\mms_i\le\tau$.

Finally, $\alpha$-GMMS implies $\alpha$-MMS by applying the groupwise condition to the full set of agents. This completes the proof.
\end{proof}

For a positive integer $c$, say that $A$ is $\alpha$-EF$c$ if, for every ordered pair $i,j$, there is a set $R_{ij}\subseteq A_j$ with $|R_{ij}|\le c$ such that $v_i(A_i)\ge\alpha v_i(A_j\setminus R_{ij})$.

\begin{lemma}
\label{lem:efc-propk}
Fix $\alpha\in(0,1]$ and a positive integer $c$. Every $\alpha$-EF$c$ allocation is $\alpha$-PROP${c(n-1)}$.
\end{lemma}

\begin{proof}
Fix an agent $i$. For each $j\ne i$, choose a set $R_{ij}\subseteq A_j$ with $|R_{ij}|\le c$ and $v_i(A_i)\ge\alpha v_i(A_j\setminus R_{ij})$.
Let $S_i:=\bigcup_{j\ne i}R_{ij}$.
Then $S_i\subseteq G\setminus A_i$ and $|S_i|\le c(n-1)$. Also,
\[
    v_i(G) =v_i(A_i)+\sum_{j\ne i}v_i(A_j) \le
    v_i(A_i)+\sum_{j\ne i}
    \left(\frac{v_i(A_i)}{\alpha}+v_i(R_{ij})\right) =v_i(A_i)+\frac{(n-1)v_i(A_i)}{\alpha}+v_i(S_i).
\]
Therefore
\[
    v_i(A_i\cup S_i)-\alpha\frac{v_i(G)}{n}
    \ge
    \frac{1-\alpha}{n}v_i(A_i)
    +
    \left(1-\frac{\alpha}{n}\right)v_i(S_i)
    \ge0.
\]
Thus agent $i$ satisfies $\alpha$-PROP${c(n-1)}$.
\end{proof}

For PMMS, let $\mu_i^2(S)
    :=
    \max_{(B_1,B_2)\in\Pi_2(S)}
    \min\{v_i(B_1),v_i(B_2)\}$.
An allocation is $\alpha$-PMMS if
$v_i(A_i)\ge\alpha\mu_i^2(A_i\cup A_j)$ for every ordered pair of
distinct agents $i,j\in N$~\cite{caragiannis2019unreasonable}.

\begin{lemma}
\label{lem:pmms-prop1}
Fix $\alpha\in(0,1]$. Every $\alpha$-PMMS allocation is $\frac{\alpha n}{\alpha+2(n-1)}$-PROP1.
\end{lemma}

\begin{proof}
Fix an agent $i$.

We first use a simple two-partition bound. If a set of goods has total value $W$ and every good has value at most $q$, then its two-way maximin share is at least $(W-q)/2$. To see this, add goods to one bundle until its value first reaches $(W-q)/2$. The value of that bundle is at most $(W+q)/2$, so the other bundle also has value at least $(W-q)/2$.

Apply this bound to $A_i\cup A_j$. Its total value is $v_i(A_i)+v_i(A_j)$, and every good in it has value at most $v_i(A_i)+h_i$. Hence
\[
    \mu_i^2(A_i\cup A_j)
    \ge
    \frac{v_i(A_j)-h_i}{2}.
\]
The $\alpha$-PMMS condition gives us $v_i(A_i)\ge\frac{\alpha}{2}(v_i(A_j)-h_i)$,
so we have that $v_i(A_j)\le\frac{2v_i(A_i)}{\alpha}+h_i$.
Summing over $j\ne i$, we get
\[
    v_i(G) =v_i(A_i)+\sum_{j\ne i}v_i(A_j) \le v_i(A_i)+(n-1)\left(\frac{2v_i(A_i)}{\alpha}+h_i\right) \le\left(1+\frac{2(n-1)}{\alpha}\right)(v_i(A_i)+h_i).
\]
Therefore,
\[
    v_i(A_i)+h_i
    \ge
    \frac{\alpha n}{\alpha+2(n-1)}\frac{v_i(G)}{n}.
\]
Equation \eqref{eq:prop1-max-good-form} proves the claim.
\end{proof}

\begin{corollary}
\label{cor:other-notions}
Fix $n\ge2$. For every notion covered by
\Cref{lem:prop1-implications,lem:efc-propk,lem:pmms-prop1}, with $c$
fixed for EF$c$, no randomized online algorithm can guarantee any positive multiplicative approximation against an adaptive adversary. The same conclusion holds for any simultaneous guarantee that includes one of these notions.
\end{corollary}

\begin{proof}
A positive approximation to any listed notion would, by the preceding
lemmas, give $\beta$-PROP$k$ for some $\beta>0$ and some fixed $k$.
This contradicts the randomized extension of \Cref{thm:online-propk-impossibility}. The same argument applies to a simultaneous guarantee containing one of the listed
notions.
\end{proof}

\section{Chores: Inapproximability and Consequences for Other Fairness Notions}
\label{app:chores}

The main text and Appendix~\ref{app:fairness-implications} consider goods. In this appendix, we use the same sets of agents and items, the same arrival order, and the same online model as in \Cref{sec:preliminaries}, but each item is a \emph{chore}. Agent $i$ has a nonnegative additive cost function $c_i:2^G\to\mathbb{R}_{\ge0}$. Thus $c_i(S)=\sum_{g\in S}c_i(g)$ for every $S\subseteq G$.

For a nonnegative integer $k$ and a set $X\subseteq G$, let
\begin{equation}
\label{eq:chores-residual-cost}
    \widehat c_i^k(X)
    :=
    \min_{\substack{S\subseteq X\\ |S|\le k}}
    c_i(X\setminus S).
\end{equation}
This is the least cost remaining after removing at most $k$ chores from $X$. For chores, a smaller approximation factor is stronger, and we therefore take $\lambda\ge1$.

\begin{definition}[$\lambda$-PROP$k$ for chores]
\label{def:chores-propk}
    Fix $\lambda\ge1$ and an integer $k\ge1$. An allocation $A$ is $\lambda$-PROP$k$ for chores if $\widehat c_i^k(A_i) \le \lambda\frac{c_i(G)}{n}$ for every $i \in N$.
\end{definition}

Every allocation is $n$-PROP$k$, since $\widehat c_i^k(A_i)\le c_i(A_i)\le c_i(G)$. The next theorem shows that no smaller universal factor is possible online, even for a fixed $k$.

\subsection{The \texorpdfstring{PROP$k$}{PROPk} Lower Bound}

We will use the following monotonicity property of \eqref{eq:chores-residual-cost}: if $X\subseteq Y$, then $\widehat c_i^k(X)\le \widehat c_i^k(Y)$. Indeed, for every $S\subseteq Y$ with $|S|\le k$, nonnegativity and additivity give
\[
    c_i(Y\setminus S)
    \ge
    c_i(X\setminus(S\cap X))
    \ge
    \widehat c_i^k(X).
\]
Taking the minimum over all such $S$ gives the stated inequality.

Fix $k\geq 1$ for the rest of this subsection. The induction below
uses the same broad recursive idea as the online MMS lower bounds of
Seddighin and Seddighin~\cite{seddighin2025lowerbound} and Song et
al.~\cite{song2025onlinemmschores}: an $r$-agent construction is repeated for the first $r$ agents, while the costs of the additional agent increase geometrically. The new inductive statement controls
the cost left after removing any $k$ chores. It therefore gives a
PROP$k$ lower bound rather than only an MMS lower bound.

More precisely, the next lemma constructs an adaptive adversary such that, regardless of the recipients, at some prefix one agent bears almost all of her total cost for the chores seen so far even after removing up to $k$ chores from her bundle. Each new copy of the $r$-agent construction is scaled so that earlier chores contribute little to the relevant agent's total cost.

\begin{lemma}
\label{lem:chores-concentration}
For every positive integer $r$ and every $\eps\in(0,1)$, there is an adaptive adversary for $r$ agents together with a positive integer $T(r,\eps)$ such that, for every possible sequence of recipients, some prefix $G'$ of at most $T(r,\eps)$ chores and some agent $i\in[r]$ satisfy
\begin{equation}
\label{eq:chores-concentration}
    \widehat c_i^k(A_i)>(1-\eps)c_i(G'),
\end{equation}
where $A_i$ is agent $i$'s bundle at that prefix. Thus, the cost remaining after removing up to $k$ chores from $A_i$ exceeds a $(1-\eps)$ fraction of agent $i$'s cost for all chores in $G'$. The adversary chooses each cost vector using only the recipients of earlier chores. Its first chore has cost one for every agent, and every chore it presents has positive cost for every agent.
\end{lemma}

\begin{proof}
We proceed by induction on $r$.

For $r=1$, set $T(1,\eps):=\left\lfloor k/\eps\right\rfloor+1$ and present $T(1,\eps)$ chores, each of cost one. The only agent receives every chore. Since $\eps T(1,\eps)>k$,
\[
    \widehat c_1^k(A_1)
    =
    T(1,\eps)-k
    >
    (1-\eps)T(1,\eps)
    =
    (1-\eps)c_1(G').
\]

Now assume the statement for $r$ agents, and consider agents $1,\dots,r+1$. Set
\[
    \eta:=\frac{\eps}{4},
    \quad
    T_0:=T(r,\eta),
    \quad
    q:=1+\frac{1}{\eta},
    \quad
    b_s:=q^{s-1}\quad(s=1,\dots,T_0).
\]
The strategy proceeds in blocks. At the beginning of a block, let $P_i$ be agent $i$'s total cost for all preceding chores, for each $i\in[r]$, and set $w_i:=\max\{1, P_i/\eta\}$.
Start a new copy of the $r$-agent strategy with parameter $\eta$. If its $s$th chore has costs $(d_1,\dots,d_r)$, present an actual chore with costs $(w_1d_1,\dots,w_rd_r,b_s)$.
If this chore is assigned to agent $r+1$, end the current block and start a new one. Otherwise, pass its recipient to the $r$-agent strategy and continue the current block.

Suppose first that the $r$-agent strategy reaches \eqref{eq:chores-concentration} within a block before agent $r+1$ receives a chore. Let $i\in[r]$ be the agent given by the induction hypothesis, let $X_i$ be the chores allocated to $i$ in the current block up to that point, and let $L_i$ be agent $i$'s total cost for all chores in that part of the block. Scaling agent $i$'s costs by $w_i$ gives
\[
    \widehat c_i^k(X_i)>(1-\eta)L_i.
\]
The first chore of the $r$-agent strategy has cost one to every participating agent, so $L_i\ge w_i$. By the definition of $w_i$, $P_i\le\eta w_i\le\eta L_i$. Since $X_i\subseteq A_i$, the monotonicity shown above gives
\[
\begin{aligned}
    \frac{\widehat c_i^k(A_i)}{c_i(G')} \ge
    \frac{\widehat c_i^k(X_i)}{P_i+L_i} >
    \frac{1-\eta}{1+\eta}
    >1-\eps.
\end{aligned}
\]
The last inequality follows from $\eta=\eps/4$. Thus \eqref{eq:chores-concentration} holds. Consequently, unless the desired inequality has already been reached, agent $r+1$ receives a chore within the first $T_0$ positions of every block.

It remains to consider the case in which agent $r+1$ repeatedly ends the blocks. If the first $s-1$ chores of a block are allocated to agents in $[r]$ and the $s$th chore is allocated to agent $r+1$, then the cost missed by agent $r+1$ in that block is
\begin{equation}
\label{eq:chores-geometric-sum}
    \sum_{h=1}^{s-1}b_h
    = \frac{q^{s-1}-1}{q-1}
    = \eta(b_s-1)
    < \eta b_s.
\end{equation}
Let $Q:= \lfloor k b_{T_0}/\eta \rfloor+1$.
Complete $Q$ blocks, unless \eqref{eq:chores-concentration} has already been reached for an agent in $[r]$. Let $L$ be agent $r+1$'s total cost for the $Q$ chores she receives, and let $O$ be her total cost for the chores allocated to the other agents. Summing \eqref{eq:chores-geometric-sum} over the blocks gives $O<\eta L$. Also, every received chore has cost at most $b_{T_0}$ and at least one, so $L\ge Q>\frac{k b_{T_0}}{\eta}$.
Removing at most $k$ chores can remove cost at most $k b_{T_0}$. Therefore
\[
    \widehat c_{r+1}^k(A_{r+1})
    \ge L-kb_{T_0}
    >(1-\eta)L.
\]
It follows that
\[
    \frac{\widehat c_{r+1}^k(A_{r+1})}{c_{r+1}(G')}
    >
    \frac{1-\eta}{1+\eta}
    >1-\eps.
\]
Thus the induction closes with $T(r+1,\eps):=QT_0$. At the beginning of the first block, every $P_i$ is zero and every $w_i$ is one, so the first chore has cost one for all $r+1$ agents. All other costs in the construction are positive as well.
\end{proof}

\begin{theorem}
\label{thm:chores-propk-impossibility}
Fix $n\ge2$, $k\ge1$, and $\lambda\in[1,n)$. For every (deterministic or randomized) online algorithm, there is an adaptive adversary under which the final allocation is not $\lambda$-PROP$k$; for a randomized algorithm, this holds with probability one. The statement remains true even when both of the following hold:
\begin{enumerate}[(i)]
    \item the final number of chores is known before the first chore arrives; and
    \item every cost lies in $[0,1]$.
\end{enumerate}
The adversary observes only the realized recipients of earlier chores.
\end{theorem}

\begin{proof}
Set $\eps:=\frac{n-\lambda}{2n}$.
Then $1-\eps>\lambda/n$. Apply \Cref{lem:chores-concentration} with $r=n$. Along every sequence of recipients, some prefix $G'$ and some agent $i$ satisfy
\[
    \widehat c_i^k(A_i)
    >(1-\eps)c_i(G')
    >\lambda\frac{c_i(G')}{n}.
\]
The adversary ends the instance at this prefix. The allocation is therefore not $\lambda$-PROP$k$. Since the argument applies to every sequence of recipients, it applies to every realization of a randomized algorithm's private choices.

The bound $T(n,\eps)$ in \Cref{lem:chores-concentration} depends only on $n$, $k$, and $\lambda$. The adversary may reveal this number before the first chore arrives and, after the displayed inequality is reached, add zero-cost chores until exactly $T(n,\eps)$ chores have arrived. These added chores do not change the inequality.

Finally, consider the finite tree containing every recipient history of length at most $T(n,\eps)$ under the strategy. Only finitely many cost vectors occur in this tree. Multiplying all of them by one common positive constant makes every cost lie in $[0,1]$. Each inequality above is homogeneous, so the same conclusion holds after this scaling.
\end{proof}

\subsection{Consequences for Proportionality and Envy}

We next use direct multiplicative cost versions of standard chore fairness notions \cite{chen2026competitiveanalysis,garg2026relations}. An allocation is $\lambda$-PROP if
\[
    c_i(A_i)\le\lambda\frac{c_i(G)}{n}
    \quad\text{for every }i.
\]
It is $\lambda$-PROP1 when \Cref{def:chores-propk} holds with $k=1$. It is $\lambda$-PROPX if, for every agent $i$ and every chore $g\in A_i$ with $c_i(g)>0$,
\[
    c_i(A_i\setminus\{g\})
    \le
    \lambda\frac{c_i(G)}{n};
\]
the condition is automatic when $A_i$ contains no positive-cost chore. Under additive chore costs, the standard definitions of PROPm and PROPavg reduce to PROPX \cite{garg2026relations}; their direct multiplicative versions therefore coincide with $\lambda$-PROPX. Hence
\begin{equation}
\label{eq:chores-prop-implications}
    \lambda\text{-PROP}
    \Longrightarrow
    \lambda\text{-PROPX}
    \Longrightarrow
    \lambda\text{-PROP1}.
\end{equation}
The same conclusion holds for $\lambda$-PROPm and $\lambda$-PROPavg.

For a nonnegative integer $c$, an allocation is $\lambda$-EF$c$ for chores if, for every ordered pair $i,j$, there is a set $S_{ij}\subseteq A_i$ with $|S_{ij}|\le c$ such that
\begin{equation}
\label{eq:chores-efc}
    c_i(A_i\setminus S_{ij})\le\lambda c_i(A_j).
\end{equation}
Thus $\lambda$-EF is $\lambda$-EF$0$, and $\lambda$-EF1 is the case $c=1$. An allocation is $\lambda$-EFX if, for every ordered pair $i,j$ and every $g\in A_i$ with $c_i(g)>0$,
\[
    c_i(A_i\setminus\{g\})\le\lambda c_i(A_j).
\]
In particular, $\lambda$-EFX implies $\lambda$-EF1. An allocation is epistemic $\lambda$-EF$c$ if, for every agent $i$, there is a complete allocation $B^i$ with $B_i^i=A_i$ in which agent $i$ satisfies $\lambda$-EF$c$. The cases $c=0$ and $c=1$ are denoted EEF and EEF1, respectively; epistemic EFX is denoted EEFX.

Define
\begin{equation}
\label{eq:chores-theta}
    \theta_n(\lambda):=\frac{n\lambda}{\lambda+n-1}.
\end{equation}
For every finite $\lambda\ge1$, we have $\theta_n(\lambda)<n$.

\begin{lemma}
\label{lem:chores-envy-propk}
Fix a nonnegative integer $c$ and a finite $\lambda\ge1$. Every $\lambda$-EF$c$ allocation and every epistemic $\lambda$-EF$c$ allocation is $\theta_n(\lambda)$-PROP$c$, where PROP$0$ means PROP. In particular, every $\lambda$-EFX and every $\lambda$-EEFX allocation is $\theta_n(\lambda)$-PROP1.
\end{lemma}

\begin{proof}
Fix an agent $i$ and set $r_i:=\widehat c_i^c(A_i)$. For ordinary EF$c$, take $B^i=A$. For epistemic EF$c$, take the allocation $B^i$ from the definition. For each $j\ne i$, \eqref{eq:chores-efc} and the definition of $r_i$ give
\[
    r_i\le\lambda c_i(B_j^i),
\]
so $c_i(B_j^i)\ge r_i/\lambda$. Also $c_i(A_i)\ge r_i$. Therefore
\[
    c_i(G)
    =c_i(A_i)+\sum_{j\ne i}c_i(B_j^i)
    \ge
    r_i+(n-1)\frac{r_i}{\lambda}.
\]
Rearranging gives
\[
    r_i
    \le
    \frac{\lambda}{\lambda+n-1}c_i(G)
    =
    \theta_n(\lambda)\frac{c_i(G)}{n}.
\]
For EFX and EEFX, use the implication to EF1 and take $c=1$.
\end{proof}

When $c=0$, $\theta_n(\lambda)$-PROP is stronger than $\theta_n(\lambda)$-PROP1, so \Cref{thm:chores-propk-impossibility} applies in this case as well. It follows from \Cref{thm:chores-propk-impossibility} and \Cref{lem:chores-envy-propk} that no randomized online algorithm guarantees any finite multiplicative approximation to EF, EF1, EFX, EEF, EEF1, EEFX, EF$c$, or epistemic EF$c$ for any fixed $c$.

\subsection{Consequences for Share-Based Notions}

For a set of chores $H\subseteq G$ and a positive integer $q$, using the partition notation $\Pi_q(H)$ from Appendix~\ref{app:fairness-implications}, define
\begin{equation}
\label{eq:chores-minmax-share}
    \mu_i^{\mathrm{ch}}(H,q)
    :=
    \min_{(X_1,\dots,X_q)\in\Pi_q(H)}
    \max_{t\in[q]}c_i(X_t).
\end{equation}
For the AnyPrice share for chores \cite{babaioff2023fairshare,garg2026relations}, define $\mathrm{APS}_i^{\mathrm{ch}}(\varnothing,q):=0$, and, for $H\ne\varnothing$,
\begin{equation}
\label{eq:chores-aps}
    \mathrm{APS}_i^{\mathrm{ch}}(H,q)
    :=
    \max_{\substack{x\in\mathbb{R}_{\ge0}^{H}\\ \sum_{g\in H}x_g=1}}
    \min_{\substack{Y\subseteq H\\ \sum_{g\in Y}x_g\ge1/q}}
    c_i(Y).
\end{equation}
Agent $i$'s chore maximin share and AnyPrice share are $\mms_i^{\mathrm{ch}}:=\mu_i^{\mathrm{ch}}(G,n)$ and $\mathrm{APS}_i^{\mathrm{ch}}:=\mathrm{APS}_i^{\mathrm{ch}}(G,n)$. An allocation is $\lambda$-MMS or $\lambda$-APS if, respectively,
\[
    c_i(A_i)\le\lambda\mms_i^{\mathrm{ch}}
    \quad\text{or}\quad
    c_i(A_i)\le\lambda\mathrm{APS}_i^{\mathrm{ch}}
    \quad\text{for every }i.
\]
It is $\lambda$-PMMS or $\lambda$-PAPS if, respectively,
\[
    c_i(A_i)
    \le
    \lambda\mu_i^{\mathrm{ch}}(A_i\cup A_j,2)
    \quad\text{or}\quad
    c_i(A_i)
    \le
    \lambda\mathrm{APS}_i^{\mathrm{ch}}(A_i\cup A_j,2)
\]
for every $i\ne j$. It is $\lambda$-GMMS (respectively, $\lambda$-GAPS) if, for every nonempty $S\subseteq N$ and every $i\in S$, the first (respectively, second) inequality below holds:
\[
    c_i(A_i)
    \le
    \lambda\mu_i^{\mathrm{ch}}\left(\bigcup_{j\in S}A_j,|S|\right),
    \quad
    c_i(A_i)
    \le
    \lambda\mathrm{APS}_i^{\mathrm{ch}}\left(\bigcup_{j\in S}A_j,|S|\right).
\]

Let $\mathcal E_i^{\mathrm{ch}}$ and $\mathcal E_i^{1,\mathrm{ch}}$ be the sets of allocations in which agent $i$ satisfies exact EFX and exact EF1 for chores, respectively. Both sets are nonempty because giving agent $i$ the empty bundle makes her EFX-satisfied. Define
\begin{equation}
\label{eq:chores-mxs}
    \mathrm{MXS}_i^{\mathrm{ch}}
    :=
    \max_{B\in\mathcal E_i^{\mathrm{ch}}}c_i(B_i),
    \quad
    \mathrm{M1S}_i^{\mathrm{ch}}
    :=
    \max_{B\in\mathcal E_i^{1,\mathrm{ch}}}c_i(B_i).
\end{equation}
These are the cost forms of the minimum EFX share and minimum EF1 share: utilities for chores are $-c_i$, so minimizing utility over the relevant allocations is equivalent to maximizing cost. An allocation is $\lambda$-MXS or $\lambda$-M1S if, respectively,
\[
    c_i(A_i)\le\lambda\mathrm{MXS}_i^{\mathrm{ch}}
    \quad\text{or}\quad
    c_i(A_i)\le\lambda\mathrm{M1S}_i^{\mathrm{ch}}
    \quad\text{for every }i.
\]

For an agent $i$, let
\begin{equation}
\label{eq:chores-common-upper-bound}
    U_i
    :=
    \frac{c_i(G)}{n}
    +
    \left(1-\frac1n\right)\max_{g\in G}c_i(g).
\end{equation}

\begin{lemma}
\label{lem:chores-share-upper-bound}
For every agent $i$,
\[
    \mms_i^{\mathrm{ch}}\le U_i,
    \quad
    \mathrm{APS}_i^{\mathrm{ch}}\le U_i,
    \quad
    \mathrm{MXS}_i^{\mathrm{ch}}\le U_i,
    \quad
    \mathrm{M1S}_i^{\mathrm{ch}}\le U_i.
\]
\end{lemma}

\begin{proof}
For MMS, assign the chores one at a time to a currently least costly one of $n$ bundles, according to $c_i$. Consider a bundle of maximum final cost, and let $g$ be the last chore assigned to it. Immediately before $g$ was assigned, that bundle had cost at most the average current cost, and hence at most $(c_i(G)-c_i(g))/n$. Its final cost is therefore at most
\[
    \frac{c_i(G)-c_i(g)}{n}+c_i(g)
    \le U_i.
\]
This gives an $n$-partition whose maximum cost is at most $U_i$, and hence $\mms_i^{\mathrm{ch}}\le U_i$.

For APS, fix a nonnegative price vector $x$ with total price one, and take an MMS-optimal partition $(X_1,\dots,X_n)$. Some bundle $X_j$ has price at least $1/n$, and every bundle has cost at most $\mms_i^{\mathrm{ch}}$. Hence
\[
    \min_{\substack{Y\subseteq G\\ \sum_{g\in Y}x_g\ge1/n}}c_i(Y)
    \le c_i(X_j)
    \le \mms_i^{\mathrm{ch}}.
\]
Taking the maximum over $x$ gives $\mathrm{APS}_i^{\mathrm{ch}}\le\mms_i^{\mathrm{ch}}\le U_i$.

For MXS and M1S, fix $B$ in $\mathcal E_i^{\mathrm{ch}}$ or $\mathcal E_i^{1,\mathrm{ch}}$, and let $q_i:=\max_{g\in B_i}c_i(g)$, with $q_i=0$ when $B_i=\varnothing$. In either case, agent $i$ satisfies EF1, so
\[
    c_i(B_i)-q_i\le c_i(B_j)
    \quad\text{for every }j\ne i.
\]
Consequently,
\[
    c_i(G)
    =c_i(B_i)+\sum_{j\ne i}c_i(B_j)
    \ge n c_i(B_i)-(n-1)q_i.
\]
Since $q_i\le\max_{g\in G}c_i(g)$, rearranging gives $c_i(B_i)\le U_i$. Taking the maximum over the appropriate set of allocations proves both bounds.
\end{proof}

Define $\psi_n(\lambda) := \frac{\lambda n^2}{n+\lambda(n-1)}$.
If $1\le\lambda<n$, then $\psi_n(\lambda)<n$.

\begin{lemma}
\label{lem:chores-share-prop1}
Fix $1\le\lambda<n$. Every $\lambda$-MMS, $\lambda$-APS, $\lambda$-MXS, or $\lambda$-M1S allocation is $\psi_n(\lambda)$-PROP1. The same conclusion holds for $\lambda$-GMMS and $\lambda$-GAPS.
\end{lemma}

\begin{proof}
Fix an agent $i$ and set $r_i:=\widehat c_i^1(A_i)$. By \Cref{lem:chores-share-upper-bound}, each of the $\lambda$-MMS, $\lambda$-APS, $\lambda$-MXS, and $\lambda$-M1S conditions gives
\begin{equation}
\label{eq:chores-own-cost-upper-bound}
    c_i(A_i)\le\lambda U_i.
\end{equation}
We also have
\begin{equation}
\label{eq:chores-residual-global-max}
    r_i
    \le
    c_i(G)-\max_{g\in G}c_i(g).
\end{equation}
Indeed, if a maximum-cost chore belongs to $A_i$, remove it; otherwise, $c_i(A_i)\le c_i(G)-\max_{g\in G}c_i(g)$ even without removing a chore.

If $c_i(G)=0$, the claim is immediate. Otherwise, set $x:=\max_{g\in G}c_i(g)/c_i(G)$. From \eqref{eq:chores-own-cost-upper-bound}--\eqref{eq:chores-residual-global-max},
\[
    \frac{r_i}{c_i(G)}
    \le
    \min\left\{
        \frac{\lambda}{n}(1+(n-1)x),
        1-x
    \right\}.
\]
The first expression increases with $x$, and the second decreases. They are equal at $x=\frac{n-\lambda}{n+\lambda(n-1)}$,
where their common value is $\frac{\lambda n}{n+\lambda(n-1)}$.
Therefore
\[
    r_i
    \le
    \psi_n(\lambda)\frac{c_i(G)}{n}.
\]
Finally, $\lambda$-GMMS implies $\lambda$-MMS, and $\lambda$-GAPS implies $\lambda$-APS, by taking $S=N$.
\end{proof}

Unlike the goods setting, exact MMS need not imply exact PROP1 for chores \cite{garg2026relations}. The quantitative implication in \Cref{lem:chores-share-prop1} is enough here because $\psi_n(\lambda)<n$ whenever $\lambda<n$. Thus no randomized online algorithm guarantees $\lambda$-MMS, $\lambda$-APS, $\lambda$-MXS, $\lambda$-M1S, $\lambda$-GMMS, or $\lambda$-GAPS for any $1\le\lambda<n$. The MMS conclusion agrees with the independent $(n-\eps)$ online lower bound of Song et al.~\cite{song2025onlinemmschores}.

PMMS has a different universal factor, so we treat it separately. We use the following elementary bound: if a set $H$ has total cost $W$ and maximum single-chore cost $p$, then
\begin{equation}
\label{eq:chores-two-partition-bound}
    \mu_i^{\mathrm{ch}}(H,2)\le\frac{W+p}{2}.
\end{equation}
To see this, assign each chore to a currently less costly one of two bundles. If $g$ is the last chore assigned to a final maximum-cost bundle, its cost before receiving $g$ is at most $(W-c_i(g))/2$.

\begin{lemma}
\label{lem:chores-pmms-prop1}
Fix $1\le\lambda<2$.
\begin{enumerate}[(i)]
    \item If $n=2$, every $\lambda$-PMMS allocation is $\frac{4\lambda}{2+\lambda}$-PROP1.
    \item If $n\ge3$, every $\lambda$-PMMS allocation is $\frac{n\lambda}{2}$-PROP1.
\end{enumerate}
In either case, the resulting PROP1 factor is smaller than $n$.
\end{lemma}

\begin{proof}
When $n=2$, PMMS is MMS on the full set of chores, and the first statement is \Cref{lem:chores-share-prop1} with $n=2$.

Suppose $n\ge3$. Fix an agent $i$, and define $s:=c_i(A_i)$, $q_i:=\max_{g\in A_i}c_i(g)$, and $a_j:=c_i(A_j)$ for $j \neq i$,
where $q_i=0$ when $A_i=\varnothing$.
Set $d:=(2-\lambda)/\lambda$. For each $j\ne i$, the largest cost of a chore in $A_i\cup A_j$ is at most $q_i+a_j$. By \eqref{eq:chores-two-partition-bound} and $\lambda$-PMMS, $s \le \frac{\lambda}{2}(s+q_i+2a_j)$,
and hence $a_j\ge\frac{ds-q_i}{2}$.
Since every $a_j$ is nonnegative and there are at least two agents other than $i$,
\begin{equation}
\label{eq:chores-pmms-total-others}
    \sum_{j\ne i}a_j\ge (ds-q_i)_+.
\end{equation}

If $q_i\ge ds$, then
\[
    \widehat c_i^1(A_i)
    =s-q_i
    \le(1-d)s
    \le\frac{\lambda}{2}s
    \le\frac{\lambda}{2}c_i(G),
\]
where $1-d\le\lambda/2$ is equivalent to $(\lambda-2)^2\ge0$. If $q_i<ds$, then \eqref{eq:chores-pmms-total-others} gives
\[
    c_i(G)
    \ge
    s+ds-q_i
    =
    \frac{2s}{\lambda}-q_i.
\]
Therefore
\[
    \frac{\lambda}{2}c_i(G)
    \ge
    s-\frac{\lambda q_i}{2}
    \ge
    s-q_i
    =\widehat c_i^1(A_i),
\]
where the second inequality uses $\lambda\le2$. Thus in all cases
\[
    \widehat c_i^1(A_i)
    \le
    \frac{\lambda}{2}c_i(G)
    =
    \frac{n\lambda}{2}\frac{c_i(G)}{n}. \qedhere
\]
\end{proof}

Combining \Cref{thm:chores-propk-impossibility} and \Cref{lem:chores-pmms-prop1}, no randomized online algorithm guarantees $\lambda$-PMMS for any $1\le\lambda<2$. The same conclusion holds for $\lambda$-PAPS because $\mathrm{APS}_i^{\mathrm{ch}}(H,2)\le\mu_i^{\mathrm{ch}}(H,2)$ for every set $H$: for every price vector, one bundle in a two-way min-max partition has price at least $1/2$ and cost at most $\mu_i^{\mathrm{ch}}(H,2)$. Thus $\lambda$-PAPS implies $\lambda$-PMMS.

An allocation is \emph{$\lambda$-pairwise proportional} ($\lambda$-PPROP) if $c_i(A_i)\le\frac{\lambda}{2}c_i(A_i\cup A_j)$ for every $i \neq j$.
Since $\mu_i^{\mathrm{ch}}(A_i\cup A_j,2)\ge c_i(A_i\cup A_j)/2$, $\lambda$-PPROP also implies $\lambda$-PMMS. Hence no randomized online algorithm guarantees $\lambda$-PAPS or $\lambda$-PPROP for any $1\le\lambda<2$.

\begin{corollary}
\label{cor:chores-other-notions}
Against an adaptive adversary, the following statements hold for randomized online algorithms.
\begin{enumerate}[(i)]
    \item No finite multiplicative approximation is guaranteed for EF, EF1, EFX, EEF, EEF1, EEFX, EF$c$, or epistemic EF$c$ for any fixed $c$.
    \item For every $1\le\lambda<n$, no $\lambda$-approximation is guaranteed for PROP, PROP1, PROP$k$ for any fixed $k$, PROPX, PROPm, PROPavg, MMS, APS, MXS, M1S, GMMS, or GAPS.
    \item For every $1\le\lambda<2$, no $\lambda$-approximation is guaranteed for PMMS, PAPS, or PPROP.
\end{enumerate}
The same statements hold for any simultaneous guarantee containing one of these notions. In each case, the adaptive adversary obtains the stated violation with probability one, and the known-horizon and bounded-cost conclusions of \Cref{thm:chores-propk-impossibility} apply.
\end{corollary}

The factors $n$ and $2$ in this corollary are universal. Every allocation is $n$-PROP, and hence also $n$-PROP$k$, $n$-PROPX, $n$-PROPm, and $n$-PROPavg. Also, $\mms_i^{\mathrm{ch}}\ge c_i(G)/n$, so every allocation is $n$-MMS. For additive chores, MMS implies MXS \cite{garg2026relations}, equivalently $\mms_i^{\mathrm{ch}}\le\mathrm{MXS}_i^{\mathrm{ch}}$, and $\mathrm{MXS}_i^{\mathrm{ch}}\le\mathrm{M1S}_i^{\mathrm{ch}}$. Thus every allocation is also $n$-MXS and $n$-M1S. The same averaging argument within every group shows that every allocation is $n$-GMMS.

For APS, if $c_i(H)>0$, set $x_g:=c_i(g)/c_i(H)$ in \eqref{eq:chores-aps}; then every set of price at least $1/q$ has cost at least $c_i(H)/q$. The case $c_i(H)=0$ is immediate. Therefore
\[
    \mathrm{APS}_i^{\mathrm{ch}}(H,q)\ge\frac{c_i(H)}{q}.
\]
It follows that every allocation is $n$-APS and $n$-GAPS, and every allocation is $2$-PAPS. Finally,
\[
    \mu_i^{\mathrm{ch}}(A_i\cup A_j,2)
    \ge
    \frac{c_i(A_i\cup A_j)}{2}
    \ge
    \frac{c_i(A_i)}{2},
\]
so every allocation is $2$-PMMS; every allocation is also $2$-PPROP directly from its definition.

\begin{remark}
\label{rem:chores-sparse-costs}
The restriction in \Cref{thm:online-propk-impossibility} that every good is positively valued by at most two agents has no analogue here when $n\ge3$. If every chore has positive cost for at most two agents, then every arriving chore has a zero-cost agent. Assigning each chore to such an agent gives every agent zero cost for her own bundle and satisfies exact proportionality and envy-freeness.
\end{remark}

\section{Omitted Proofs from  \texorpdfstring{\Cref{sec:miv}}{Section~\ref*{sec:miv}}}

\subsection{Comparing the Potential Changes}
\label{app:miv-comparison}

We prove the comparison used in \eqref{eq:miv-total-change-bound}. The argument uses two bounds on each increase: an absolute bound and a bound relative to the current potential term. Together, they give the total-change inequality needed for $19/30$-PROP1.

In this subsection, let $f:(0,\infty)\to(0,\infty)$ be twice continuously differentiable, strictly decreasing and strictly convex, with $f(0+)=+\infty$ and $f(\infty)=0$. Write $a:=-f'>0$. Assume that $f$ and $a$ are log-convex (that is, their logarithms are convex), that $1/a$ is convex, and that
\begin{equation}
\label{eq:miv-function-condition}
    a(x)\le f(x)(f(x)-f(x+1))
    \quad(0<x\le1).
\end{equation}
Appendix~\ref{app:miv-function} verifies these properties for the explicit function used by the algorithm. Condition \eqref{eq:miv-function-condition} relates the rate of change of $f$ to its decrease over an interval of length one, the largest normalized value of a good.

\begin{lemma}
\label{lem:miv-function-comparison}
Fix $\rho\in(0,1)$ and $M>0$. Let $S>\rho$ satisfy
\[
    f(S-\rho)-f(S+1-\rho)=M,
\]
and set $d_*:=f(S-\rho)-f(S)$. For $0\le z\le1$ and $x>\rho z$, if
\[
    f(x-\rho z)-f(x+(1-\rho)z)\le M,
\]
then
\begin{equation}
\label{eq:miv-two-comparisons}
    f(x-\rho z)-f(x)\le d_*,
    \quad
    \frac{f(x-\rho z)-f(x)}{f(x)}\le\frac{d_*}{f(S)}.
\end{equation}
Moreover,
\begin{equation}
\label{eq:miv-loss-ratio}
    \frac{d_*}{M}\le
    \begin{cases}
        \rho f(S),&S\le1,\\
        \rho f(1),&S\ge1.
    \end{cases}
\end{equation}
Finally, $f(x)>1/x$ for $0<x\le1$.
\end{lemma}

\begin{proof}
The case $z=0$ in \eqref{eq:miv-two-comparisons} is immediate. Otherwise, write $\ell:=x-\rho z$ and $u:=\ell+z$, so $x=(1-\rho)\ell+\rho u$. At fixed $\ell$, increasing $z$ increases both $f(\ell)-f(x)$ and its ratio to $f(x)$.

First suppose $\ell\ge S-\rho$. Increase $z$ to one. Convexity of $f$ implies that $f(\ell)-f(\ell+\rho)$ is nonincreasing in $\ell$. Log-convexity of $f$ implies the same for $f(\ell)/f(\ell+\rho)-1$. Evaluating at $\ell=S-\rho$ proves both bounds.

Now suppose $\ell<S-\rho$. Since $f(\ell)-f(\ell+1)>M$, increase $z$ until $f(\ell)-f(u)=M$. Along the curve defined by this equality,
\[
    u'(\ell)=\frac{a(\ell)}{a(u)},
    \quad
    x'(\ell)=(1-\rho)+\rho\frac{a(\ell)}{a(u)}>0.
\]
Convexity of $1/a$ gives
\[
    \frac1{a(x)}\le\frac{1-\rho}{a(\ell)}+\frac\rho{a(u)},
\]
and therefore
\[
    \frac{d}{d\ell}(f(\ell)-f(x))
    =-a(\ell)+a(x)\left((1-\rho)+\rho\frac{a(\ell)}{a(u)}\right)\ge0.
\]
Meanwhile $f(x)$ decreases, so the ratio $(f(\ell)-f(x))/f(x)$ also increases. Increase $\ell$ to $S-\rho$, where $u=S+1-\rho$ and $x=S$. This proves \eqref{eq:miv-two-comparisons}. The curve stays within $u-\ell\le1$: this difference increases to one because $a$ is decreasing.

For \eqref{eq:miv-loss-ratio}, first let $S\le1$. By \eqref{eq:miv-function-condition},
\begin{align*}
    \frac1{f(S)}-\frac1{f(S-\rho)}
    &=\int_{S-\rho}^{S}\frac{a(t)}{f(t)^2}\,dt\\
    &\le\int_{S-\rho}^{S}\left(1-\frac{f(t+1)}{f(t)}\right)dt\\
    &\le\rho\left(1-\frac{f(S+1-\rho)}{f(S-\rho)}\right).
\end{align*}
The last inequality holds because log-convexity of $f$ makes $f(t+1)/f(t)$ nondecreasing. Multiplying by $f(S)f(S-\rho)$ and dividing by $M$ gives $d_*/M\le\rho f(S)$.

For every $S>\rho$, the ratio in question can be written as
\[
    \frac{d_*}{M}
    =\frac{\int_0^\rho a(S-\rho+t)\,dt}
           {\int_0^1 a(S-\rho+t)\,dt}.
\]
This ratio is nonincreasing in $S$. Indeed, for $h_2>h_1$, log-convexity of $a$ makes $a(h_2+t)/a(h_1+t)$ nondecreasing in $t$. Hence, whenever $0\le t\le\rho\le u\le1$,
\[
    a(h_2+t)a(h_1+u)\le a(h_1+t)a(h_2+u).
\]
Integration over $t\in[0,\rho]$ and $u\in[\rho,1]$, followed by cross-multiplication, proves the assertion. When $S\ge1$, the ratio is thus at most its value at $S=1$, which is at most $\rho f(1)$ by the preceding case.

Finally, integrating \eqref{eq:miv-function-condition} after division by $f(t)^2$ and letting the lower endpoint tend to zero gives
\[
    \frac1{f(x)}\le\int_0^x\left(1-\frac{f(t+1)}{f(t)}\right)dt<x
    \quad(0<x\le1).
\]
The last inequality is strict since the subtracted ratio is positive. This proves the last assertion.
\end{proof}

To obtain \eqref{eq:miv-total-change-bound}, apply \eqref{eq:miv-two-comparisons} with $x=s_i^{t-1}+\rho$ and $z=z_i$. Summing the first bound gives $\sum_i d_i\le nd_*$. Summing the second and using $\Phi^{t-1}\le nf(1)$ gives $\sum_i d_i\le nf(1)d_*/f(S)$. If $S\le1$, combine the latter with the first case of \eqref{eq:miv-loss-ratio}; if $S\ge1$, combine the former with the second case. In either case, $\sum_i d_i/M\le n\rho f(1)$, as required.

\subsection{The Function Giving the \texorpdfstring{$19/30$}{19/30} Factor}
\label{app:miv-function}

We give an explicit $f$ satisfying the properties in Appendix~\ref{app:miv-comparison} and $f(1)<30/19$. Set
\begin{align}
    b&:=\frac{1431}{1000},\quad
    r:=\frac{1517}{6775},\quad
    K:=\frac{1423249}{1300800},\nonumber\\
    c&:=K\left(1+\log\frac{327121}{47720}\right)-\frac{1157}{480},\nonumber\\[1mm]
    f(x)&:=
    \begin{cases}
        \displaystyle \frac1x+c+\frac{x}{16}-\frac{25x^2}{12}
                         +\frac{875x^3}{256},&0<x\le2/5,\\[2mm]
        \displaystyle K\left(1+\log\frac{b-r}{x-r}\right),&2/5\le x\le b,\\[2mm]
        \displaystyle K\exp\left(-\frac{x-b}{b-r}\right),&x\ge b.
    \end{cases}
\label{eq:miv-explicit-function}
\end{align}
The constants make the pieces and their first two derivatives agree at both endpoints. The logarithmic part gives $1/(-f'(x))=(x-r)/K$, an affine function, so the convexity used to compare potential changes holds there with equality. The reciprocal term makes small slacks contribute large potential values; the polynomial terms let \eqref{eq:miv-function-condition} hold without doubling that reciprocal term. The exponential part makes $f(x+1)$ small enough for the same inequality.

\paragraph{Monotonicity and convexity.}
With $a=-f'$ as above,
\[
    a(x)=
    \begin{cases}
        \displaystyle x^{-2}-\frac1{16}+\frac{25x}{6}-\frac{2625x^2}{256},&0<x\le2/5,\\[2mm]
        \displaystyle \frac K{x-r},&2/5\le x\le b,\\[2mm]
        \displaystyle \frac K{b-r}\exp\left(-\frac{x-b}{b-r}\right),&x\ge b.
    \end{cases}
\]
At $x=2/5$, the first two expressions give
\[
    a(2/5)=\frac{1193}{192}=\frac K{2/5-r},
    \quad
    a'(2/5)=-\frac{6775}{192}=-\frac K{(2/5-r)^2}.
\]
The value of the first expression for $f(2/5)$ is $c+1157/480$, which agrees with the second expression by the definition of $c$. Matching at $b$ is immediate. Thus $a$ is continuously differentiable and $f$ is twice continuously differentiable.

For $x\in[0,2/5]$, define
\begin{align*}
    H(x)&:=1-\frac{x^2}{16}+\frac{25x^3}{6}-\frac{2625x^4}{256},\\
    J(x)&:=2-\frac{25x^3}{6}+\frac{2625x^4}{128},\\
    Q(x)&:=6-\frac{2625x^4}{128}.
\end{align*}
Then $a(x)=H(x)/x^2$, $-a'(x)=J(x)/x^3$, and $a''(x)=Q(x)/x^4$ for $0<x<2/5$. We claim that
\[
    \frac{99}{100}\le H(x)\le\frac{11}{10},\quad
    \frac{19}{10}\le J(x)\le\frac{271}{120},\quad
    \frac{219}{40}\le Q(x)\le6.
\]
For the lower bound on $H$, use $25/6-2625x/256\ge0$ and $x^2/16\le1/100$. Differentiating the two polynomials on the left below gives
\[
    x^3\left(\frac{25}{6}-\frac{2625x}{256}\right)
    \le\frac{25}{24}\left(\frac{32}{105}\right)^3<\frac1{10},
    \quad
    x^3\left(\frac{25}{6}-\frac{2625x}{128}\right)
    \le\frac{25}{24}\left(\frac{16}{105}\right)^3<\frac1{10}.
\]
These prove the upper bound on $H$ and the lower bound on $J$. Since
$J'(x)=x^2(-25/2+2625x/32)$ changes sign only from negative to positive, the maximum of $J$ is at an endpoint; its endpoint values are $2$ and $271/120$. The bounds on $Q$ follow directly from $0\le x\le2/5$. Consequently,
\begin{align*}
    HQ-J^2
    &\ge\frac{99}{100}\frac{219}{40}-\left(\frac{271}{120}\right)^2
      =\frac{23053}{72000}>0,\\
    2J^2-HQ
    &\ge2\left(\frac{19}{10}\right)^2-\frac{11}{10}\cdot6
      =\frac{31}{50}>0.
\end{align*}
Thus $a>0$, $a'<0$, $aa''-(a')^2>0$, and $2(a')^2-aa''>0$ on $(0,2/5)$. These are the required monotonicity of $a$, log-convexity of $a$, and convexity of $1/a$. On $[2/5,b]$, $\log a=\log K-\log(x-r)$ is convex and $1/a=(x-r)/K$ is affine. On $[b,\infty)$, $\log a$ is affine and $1/a$ is exponential. The matching of $a$ and $a'$ makes both convexity properties hold on $(0,\infty)$.

Since $f(x)=\int_x^\infty a(t)\,dt$, the function $f$ is positive, strictly decreasing and strictly convex, with the required limits. Log-convexity of $a$ and H\"older's inequality also give log-convexity of $f$: for $x,y>0$ and $\theta\in(0,1)$,
\begin{align*}
    f(\theta x+(1-\theta)y)
    &\le\int_0^\infty a(x+t)^\theta a(y+t)^{1-\theta}\,dt\\
    &\le f(x)^\theta f(y)^{1-\theta}.
\end{align*}

\paragraph{Bounds on the constants.}
The value at one is
\[
    f(1)=K\left(1+\log\frac{327121}{210320}\right).
\]
For $q\ge1$, put $z=(q-1)/(q+1)$. For any positive integer $h$,
\begin{equation}
\label{eq:miv-log-bounds}
    2\sum_{j=0}^{h-1}\frac{z^{2j+1}}{2j+1}
    \le\log q\le
    2\sum_{j=0}^{h-1}\frac{z^{2j+1}}{2j+1}
      +\frac{2z^{2h+1}}{(2h+1)(1-z^2)}.
\end{equation}
This follows by summing the positive series for $\log((1+z)/(1-z))$ and bounding its remaining terms by a geometric series. Apply \eqref{eq:miv-log-bounds} with $h=10$ to $\log(327121/210320)$ and to the two logarithms in
\[
    \log\frac{327121}{47720}=3\log2-\log\frac{381760}{327121}.
\]
Substitution in the definitions of $f(1)$ and $c$ gives the rational bounds
\begin{equation}
\label{eq:miv-constant-bounds}
    \frac{157741237}{10^8}<f(1)<\frac{157741238}{10^8},
    \quad
    \frac{78990175}{10^8}<c<\frac{78990176}{10^8}.
\end{equation}
In particular, $1<f(1)<30/19$.

\paragraph{Proving \eqref{eq:miv-function-condition} for $0<x\le2/5$.}
The terms of order $1/x^2$ cancel in $f(x)(f(x)-f(x+1))-a(x)$. We bound the remaining expression by a rational polynomial and then by a positive cubic.
Let $C_-:=7887/5000$ and $C_+:=78871/50000$, so $C_-<f(1)<C_+$ by \eqref{eq:miv-constant-bounds}. Define
\begin{align*}
    P(x)&:=\frac{7899}{10000}+\frac{x}{16}-\frac{25x^2}{12}+\frac{875x^3}{256},\\
    U(x)&:=C_+-K\sum_{j=1}^{6}\frac{(-1)^{j+1}}{j}
                           \left(\frac{x}{1-r}\right)^j,\\
    R_0(x)&:=2P(x)-U(x)+x(P(x)^2-P(x)U(x)+P'(x)).
\end{align*}
Since $x+1\le7/5<b$ and $0\le x/(1-r)<1$, the even partial sum for $\log(1+z)$ gives
\[
    f(x+1)=f(1)-K\log\left(1+\frac{x}{1-r}\right)\le U(x)\le C_+.
\]
The last inequality follows by pairing consecutive terms in the sum defining $U$. Also, $P(x)\ge7899/10000-1/3>0$. Write $f(x)=x^{-1}+\widetilde P(x)$, where $\widetilde P=P+c-7899/10000\ge P$. Direct expansion gives
\begin{align*}
    x[f(x)(f(x)-f(x+1))-a(x)]
    ={}&2\widetilde P-f(x+1)\\
       &+x(\widetilde P^2-\widetilde P f(x+1)+P').
\end{align*}
This expression decreases with $f(x+1)$. After replacing $f(x+1)$ by $U$, it increases with $\widetilde P\ge P$, since its derivative with respect to $\widetilde P$ is at least $2-xC_+>0$. Therefore
\begin{equation}
\label{eq:miv-small-polynomial}
    x[f(x)(f(x)-f(x+1))-a(x)]\ge R_0(x).
\end{equation}

Expanding $R_0$ and comparing the rational coefficients of the same powers of $x$ gives, for $x\ge0$,
\begin{align*}
    R_0(x)\ge{}&\frac{119}{50000}+\frac{39x}{40}-\frac{813x^2}{100}
      +\frac{86x^3}{5}-\frac{17x^4}{5}+\frac{117x^5}{10}\\
      &-\frac{193x^6}{10}+\frac{153x^7}{10}-\frac{17x^8}{4}-3x^{10}.
\end{align*}
For $0\le x\le2/5$, the last three terms are at least $13x^7$, because
$153/10-(17/4)x-3x^3\ge153/10-(17/4)(2/5)-3(2/5)^3>13$.
Furthermore, the identity
\begin{align*}
    &\frac25-\frac{17x}{5}+\frac{117x^2}{10}-\frac{193x^3}{10}+13x^4\\
    &\quad=\frac{3x}{125}+\left(\frac25-x\right)^2
       \left[\frac{51}{50}-\frac32\left(\frac25-x\right)
                    +13\left(\frac25-x\right)^2\right]\ge0
\end{align*}
holds because the bracket is at least $21/50$. Combining these two bounds yields
\[
    R_0(x)\ge\frac{119}{50000}+\frac{39x}{40}
                        -\frac{813x^2}{100}+\frac{84x^3}{5}>0.
\]
To justify the strict inequality, for $0\le x\le1/5$ the quadratic
$39/40-813x/100+84x^2/5$ is decreasing and is at least its value $21/1000$ at $1/5$. For $1/5\le x\le2/5$, the same cubic equals
\begin{align*}
    &\left(\frac{84x}{5}+\frac{27}{100}\right)\left(x-\frac14\right)^2
       +\frac3{50}\left(x-\frac14\right)+\frac{101}{200000}\\
    &\quad\ge\frac{363}{100}\left(x-\frac14\right)^2
       +\frac3{50}\left(x-\frac14\right)+\frac{101}{200000}
       \ge\frac{6221}{24200000}>0,
\end{align*}
where the last bound follows by completing the square. This proves \eqref{eq:miv-function-condition} on $(0,2/5]$.

\paragraph{Proving \eqref{eq:miv-function-condition} for $2/5\le x\le1$.}
Here $a(x)=K/(x-r)$. Define
\begin{equation}
\label{eq:miv-middle-polynomial}
    R_1(x):=\frac{x-r}{K}f(x)(f(x)-f(x+1))-1.
\end{equation}
We show that $R_1$ is concave on $[2/5,1]$. Its minimum is then at an endpoint, so only two values need to be bounded. Recall that $6/5<b-r<5/4$.

For $2/5\le x\le b-1$, both function values use the logarithmic expression. Differentiating gives
\[
    R_1''(x)=\frac{K}{x-r}
       \left[\frac{2(x+1-r)-f(x)/K}{(x+1-r)^2}
                    -\log\frac{x+1-r}{x-r}\right]<0.
\]
Indeed, $\log q\ge1-1/q$ for $q\ge1$ gives
\[
    \frac{f(x)}K\ge1+\frac1{b-r}>\frac95>x+1-r,
    \quad
    \log\frac{x+1-r}{x-r}>\frac1{x+1-r},
\]
which prove the sign of the displayed derivative.

For $b-1\le x\le1$, $f(x+1)$ uses the exponential expression. In the derivative calculation, write $y=(x-r)/(b-r)$, so $0<y<1$ and $y\ge1-1/(b-r)$. Then
\[
    R_1''(x)=\frac{K}{x-r}\left[
       2\log y+e^{1-1/(b-r)-y}
                 (1-y^2-(2y-y^2)\log y)\right].
\]
The coefficient of the exponential is positive. Since $e^{-u}\le1/(1+u)$ for $u\ge0$, we obtain
\[
    R_1''(x)\le\frac{K}{x-r}
       \frac{1-y^2+(y^2+2/(b-r))\log y}{y+1/(b-r)}<0.
\]
For the final sign, the first term of \eqref{eq:miv-log-bounds} gives
\[
    -\log y\ge\frac{2(1-y)}{1+y}
       >\frac{1-y^2}{y^2+2/(b-r)}.
\]
The strict inequality follows from
$2(y^2+2/(b-r))-(1+y)^2=(1-y)^2+4/(b-r)-2>0$.
The pieces of $f$ have matching derivatives, so $R_1$ is concave on the whole interval.

For the endpoint at $2/5$, \eqref{eq:miv-constant-bounds} and $\log q\le q-1$ give
\[
    f(2/5)=c+\frac{1157}{480}>\frac{16}{5},
    \quad
    f(7/5)\le K\frac{b-r}{7/5-r}<\frac98.
\]
For the endpoint at one, use $f(1)>C_-$ and the even Taylor polynomial for the exponential:
\[
    f(2)\le K\sum_{j=0}^{4}\frac1{j!}
                 \left(\frac{b-2}{b-r}\right)^j<\frac{6831}{10000}.
\]
Since $u(u-v)$ increases with $u$ whenever $2u>v$, substitution gives
\begin{align*}
    R_1(2/5)&>\frac{2/5-r}{K}\,\frac{16}{5}
                      \left(\frac{16}{5}-\frac98\right)-1>\frac1{20},\\
    R_1(1)&>\frac{1-r}{K}\,C_-
                      \left(C_--\frac{6831}{10000}\right)-1>\frac1{2000}.
\end{align*}
By concavity, $R_1(x)>0$ throughout $[2/5,1]$. This proves \eqref{eq:miv-function-condition} on the remaining interval. All hypotheses of \Cref{lem:miv-function-comparison} have now been established by algebra and calculus, completing the proof of the $19/30$-PROP1 guarantee.

\subsection{The Guarantee When Few Agents Value Each Good}
\label{app:miv-reciprocal}

We prove the $n/(n+\kappa)$ guarantee used in \Cref{thm:miv-constant}(ii). The reciprocal rule assigns each good to a recipient minimizing \eqref{eq:reciprocal-potential}.

\begin{proof}
Set $\rho:=1/(n+\kappa)$ and use $s_i^t$ and $\Psi^t$ from \eqref{eq:miv-margin} and \eqref{eq:reciprocal-potential}. Initially, $s_i^0=1-\rho$ and $\Psi^0\le n/(1-\rho)$.

Assume inductively that $s_i^{t-1}>0$ for every agent with $p_i>0$ and $\Psi^{t-1}\le\frac{n}{1-\rho}$.
This holds at time zero. Since each reciprocal term is at most the sum,
\begin{equation}
\label{eq:margin-lower-bound}
    s_i^{t-1}
    \ge
    \frac{1}{\Psi^{t-1}}
    \ge
    \frac{1-\rho}{n}
\end{equation}
for every agent with $p_i>0$.

For the current good, use $z_i$ from \eqref{eq:miv-normalized-increment} and let $Q_t:=\{i:z_i>0\}$. Only agents in $Q_t$ have changing slacks, and $|Q_t|\le\kappa$. Their updates are given by \eqref{eq:miv-slack-update}.

Equation \eqref{eq:margin-lower-bound} and $z_i\le1$ give us
\begin{equation*}
    s_i^{t-1}-\rho z_i
    \ge
    \frac{1-\rho}{n}-\rho
    =
    \frac{\kappa-1}{n(n+\kappa)}
    >0.
\end{equation*}
Thus, every possible assignment keeps all $s_i^t$ positive.

We now show that some recipient does not increase the potential. Consider the following probability distribution only for this calculation. For each $i\in Q_t$, set $\pi_i
    :=  \rho (1+(1-\rho) \cdot \frac{z_i}{s_i^{t-1}} )$.
Algebraic manipulation gives us
\begin{equation}
\label{eq:reciprocal-identity}
    \frac{1-\pi_i}{s_i^{t-1}-\rho z_i}
    +
    \frac{\pi_i}{s_i^{t-1}+(1-\rho)z_i}
    =
    \frac{1}{s_i^{t-1}}.
\end{equation}
Moreover,
\begin{align*}
\sum_{i\in Q_t}\pi_i
&=
\rho|Q_t|
+
\rho(1-\rho)\sum_{i\in Q_t}\frac{z_i}{s_i^{t-1}} \le
\rho\kappa+\rho(1-\rho)\Psi^{t-1} \le
\rho\kappa+\rho n
=1.
\end{align*}
Assign probability $\pi_i$ to each $i\in Q_t$ and distribute the remaining probability arbitrarily. For each $i\in Q_t$, the probability that $i$ receives the good is then at least $\pi_i$. Since $s_i^{t-1}+(1-\rho)z_i > s_i^{t-1}-\rho z_i >0$, the value of $1/s_i^t$ is smaller when $i$ receives the good than when it does not. Increasing the probability that $i$ receives the good above $\pi_i$ can therefore only decrease its expectation. By \eqref{eq:reciprocal-identity}, we get that
\[
    \mathbb E\left[\frac{1}{s_i^t}\right]
    \le
    \frac{1}{s_i^{t-1}}
    \quad\text{for every }i\in Q_t.
\]
For $i$ with $p_i>0$ and $i\notin Q_t$, $s_i^t=s_i^{t-1}$. Therefore, $\mathbb E[\Psi^t]\le\Psi^{t-1}$ under this distribution over recipients. At least one recipient has $\Psi^t\le\Psi^{t-1}$, and the allocation rule chooses a minimizing recipient. By induction,
\begin{equation}
\label{eq:reciprocal-nonincreasing}
    \Psi^t\le\Psi^0\le\frac{n}{1-\rho}
    \quad \text{for all } t=0,\dots,m.
\end{equation}
In particular, $s_i^m>0$ for every agent with $p_i>0$. Since $D_i^m=v_i(G)$, we get that
\[
    p_i+B_i^m
    >
    \rho v_i(G)
    =
    \frac{n}{n+\kappa}\frac{v_i(G)}{n}.
\]
As explained above, $p_i+B_i^m$ is attainable from the final bundle after adding at most one good. Thus every agent satisfies $n/(n+\kappa)$-PROP1.

The expectation bound above is obtained after fixing the realized history and the current good. It therefore remains valid when the adversary chooses later goods after seeing earlier recipients.
\end{proof}

\subsection{The Potential Decrease for the Simultaneous Guarantee}
\label{app:miv-constant-envy-step}

We prove the assertion used in \Cref{thm:miv-constant-envy}: if the current slacks are positive and $\Omega^{t-1}<1$, some recipient satisfies $\Omega^t\le\Omega^{t-1}$. Throughout, use $f$ from Appendix~\ref{app:miv-function}, $a=-f'$ as in Appendix~\ref{app:miv-comparison}, and $\rho,C,\eta$ from \Cref{subsec:miv-constant-envy}. In particular, $0<\rho<1/3$ and $0<\eta\le10^{-6}$.

\paragraph{Two estimates for one agent.}
For $0\le z\le1$ and $x>\rho z$, write
\[
    \ell:=x-\rho z,\quad u:=x+(1-\rho)z,\quad
    d:=f(\ell)-f(x),\quad \Delta:=f(\ell)-f(u).
\]
Thus $d$ is the increase when the agent does not receive the good, and $\Delta$ is the decrease obtained by changing that decision. Define
\[
    Q:=\Delta-2\eta(1+f(u)),\quad
    \tau:=f(1-\rho)-f(2-\rho).
\]
The value $\tau$ is $\Delta$ at $x=z=1$. The subtraction in $Q$ accounts for the increase in envy toward a recipient: $e^{\eta\widehat v_i(g_t)}-1\le2\eta$, and $1+f(u)$ is the recipient's multiplier after receiving the good. We will show
\begin{align}
    Q\le\tau
    &\quad\Longrightarrow\quad d\le\rho C\Delta,
    \label{eq:constant-envy-small-change}\\
    M\ge\tau,\quad Q\le M
    &\quad\Longrightarrow\quad
    d\le\frac{2001}{2000}\rho f(x)M.
    \label{eq:constant-envy-large-change}
\end{align}
The first estimate will justify uniform averaging. The second allows us to sum the increases with the weights contributed by the envy potential.

First, \Cref{lem:miv-function-comparison} gives
\begin{equation}
\label{eq:constant-envy-relative-comparison}
    \frac{d}{\rho f(x)}\le\max\{\Delta,\tau\}.
\end{equation}
Indeed, apply \eqref{eq:miv-two-comparisons} with $M=\max\{\Delta,\tau\}$; its corresponding $S$ is at most one, so \eqref{eq:miv-loss-ratio} proves the displayed bound. A second consequence is
\begin{equation}
\label{eq:constant-envy-ratio-comparison}
    \left.
    \begin{gathered}
        \rho<S\le1,\\
        \Delta\le f(S-\rho)-f(S+1-\rho)
    \end{gathered}
    \right\}
    \quad\Longrightarrow\quad d\le\rho f(S)\Delta.
\end{equation}
For $z>0$, first apply \eqref{eq:miv-two-comparisons} with $M=\Delta$. Its corresponding value of $S$ is no smaller than the one in \eqref{eq:constant-envy-ratio-comparison}. The ratio
\[
    \frac{f(S-\rho)-f(S)}{f(S-\rho)-f(S+1-\rho)}
\]
is nonincreasing on $S>\rho$, as proved in Appendix~\ref{app:miv-comparison}. Combining this monotonicity with \eqref{eq:miv-loss-ratio} proves \eqref{eq:constant-envy-ratio-comparison}. If $z=0$, then $d=\Delta=0$, and all the required estimates are immediate.

We use the following elementary bounds for the explicit function:
\begin{equation}
\label{eq:constant-envy-function-bounds}
    \begin{gathered}
        f(2499/2500)<C,\quad 1+f(1/100)<102,\\
        f(1)-f(2)>4/5,\quad a(1)-a(2)>4/5.
    \end{gathered}
\end{equation}
For the first bound, monotonicity of $a$ and the logarithmic piece of $f$ give
\[
    f(2499/2500)
    \le f(1)+\frac{K}{2500(2499/2500-r)}<\frac{789}{500},
\]
using the upper bound on $f(1)$ in \eqref{eq:miv-constant-bounds}. The second follows from the first piece of $f$ and $c<4/5$. For the last two, the endpoint estimate in Appendix~\ref{app:miv-function} gives $f(2)<7/10$. Now use $f(1)>3/2$, $a(1)=K/(1-r)>7/5$, and $a(2)=f(2)/(b-r)$ with $b-r>6/5$. All constants are those in \eqref{eq:miv-explicit-function}. Since $a$ is decreasing and $1-\rho>2/5$, these bounds also imply
\begin{equation}
\label{eq:constant-envy-threshold}
    \frac45<f(1)-f(2)\le\tau<f(2/5)<4.
\end{equation}
Finally, direct substitution into the first pieces of $f$ and $a$ gives
\begin{equation}
\label{eq:constant-envy-near-zero}
    f(v)\le\frac2v,\quad
    \frac{99}{100v^2}\le a(v)\le\frac{101}{100v^2}
    \quad(0<v\le1/100).
\end{equation}
For the derivative bounds, multiply by $v^2$ and bound the monomials of $1-v^2/16+25v^3/6-2625v^4/256$ on $[0,1/100]$; the bound on $f$ uses $c<4/5$.

Suppose first that $u\ge1/100$, so $1+f(u)<102$. The positive log-convex function $a$ is convex, and hence $a(v)-a(v+1)$ is nonincreasing. Consequently,
\begin{align*}
    &f(2499/2500-\rho)-f(2499/2500+1-\rho)-\tau\\
    &\quad=\int_{2499/2500}^1
       (a(v-\rho)-a(v+1-\rho))\,dv\\
    &\quad\ge\frac{a(1)-a(2)}{2500}
       >\frac{4}{12500}>204\eta.
\end{align*}
If $Q\le\tau$, then $\Delta\le\tau+204\eta$, so \eqref{eq:constant-envy-ratio-comparison} at $S=2499/2500$, together with $f(2499/2500)<C$, proves \eqref{eq:constant-envy-small-change}.
If $Q\le M$ and $M\ge\tau$, then \eqref{eq:constant-envy-relative-comparison} gives
\[
    \frac{d}{\rho f(x)}\le M+204\eta
       \le\frac{2001}{2000}M,
\]
where the last inequality uses $M>4/5$ and $\eta\le10^{-6}$. This proves \eqref{eq:constant-envy-large-change}.

It remains to consider $0<u<1/100$ and $z>0$. Since $x=\ell+\rho(u-\ell)$, \eqref{eq:constant-envy-near-zero} gives
\[
    \frac d\Delta
    \le\frac{101}{99}
       \frac{\int_\ell^x v^{-2}\,dv}{\int_\ell^u v^{-2}\,dv}
    =\frac{101}{99}\frac{\rho u}{x}.
\]
Using $f(x)>1/x$ and $u(1+f(u))<3$, if $Q\le M$ and $M\ge\tau$ we obtain
\[
    \frac{d}{\rho f(x)}
    \le2u\Delta
    \le2uM+4\eta u(1+f(u))
    <\frac{M}{50}+12\eta<M.
\]
This is stronger than \eqref{eq:constant-envy-large-change}.
If instead $Q\le\tau$, then \eqref{eq:constant-envy-threshold} gives
\[
    \frac{99}{100}\left(\frac u\ell-1\right)
    \le u\Delta
    \le4u+2\eta u(1+f(u))
    <\frac4{100}+6\eta<\frac7{100}.
\]
Thus $u/\ell<11/10$. By monotonicity of $a$ and \eqref{eq:constant-envy-near-zero},
\[
    \frac d{\rho\Delta}
    \le\frac{a(\ell)}{a(u)}
    \le\frac{101}{99}\left(\frac u\ell\right)^2
    <\frac{101}{99}\left(\frac{11}{10}\right)^2
    <\frac32<C.
\]
This proves \eqref{eq:constant-envy-small-change} and completes both estimates.

\paragraph{Comparing recipients.}
Fix the realized history and current good, and suppose the current slacks are positive and $\Omega^{t-1}<1$. For each agent with $p_i>0$, apply the preceding estimates with $x=s_i^{t-1}+\rho$ and $z=z_i$, adding subscript $i$ to $\ell,u,d,\Delta,Q$. Then $\Delta_i$ is exactly \eqref{eq:miv-improved-score}. For $p_i=0$, set $f(\ell_i)=f(u_i)=d_i=\Delta_i=0$ and $Q_i=-2\eta$, consistently with the convention in \Cref{subsec:miv-constant-envy}. In the envy calculation, we always use the actual normalized value $\widehat v_i(g_t)$, not $z_i$. In particular, a first maximum-valued good has $z_i=0$ but $\widehat v_i(g_t)=1$.

\emph{Case 1: $Q_i\le\tau$ for every agent.}
Since $n\rho C=(19/30)C<1$, \eqref{eq:constant-envy-small-change} implies
\[
    d_i-\frac{\Delta_i}{n}
    \le\left(\rho C-\frac1n\right)\Delta_i\le0.
\]
Under the uniform distribution over all $n$ recipients, it follows that $\mathbb E[\Phi^t]\le\Phi^{t-1}$. For distinct agents $i,j$, direct expansion gives
\begin{align*}
    &\frac{\mathbb E\!\left[e^{\eta E_{ij}^t}
                 (1+f(s_j^t+\rho))\right]}{e^{\eta E_{ij}^{t-1}}}\\
    &\quad=1+f(s_j^{t-1}+\rho)+d_j-\frac{\Delta_j}{n}\\
    &\quad\quad+\frac{e^{\eta\widehat v_i(g_t)}-1}{n}(1+f(u_j))
          +\frac{e^{-\eta\widehat v_i(g_t)}-1}{n}(1+f(\ell_j))\\
    &\quad\le\left(1+\frac{2\eta^2}{n}\right)
                   (1+f(s_j^{t-1}+\rho)).
\end{align*}
For the last inequality, use $f(\ell_j)\ge f(u_j)$, $f(u_j)\le f(s_j^{t-1}+\rho)$, and
$e^h+e^{-h}-2\le2\eta^2$ for $0\le h\le\eta\le1$. The negative coefficient of $1+f(\ell_j)$ lets us replace it by $1+f(u_j)$ in an upper bound. These observations also apply when $p_j=0$.
Summing over ordered pairs gives $\mathbb E[F_\eta^t]\le A_\eta F_\eta^{t-1}$, and therefore
\[
    \mathbb E[\Omega^t]
    \le\frac{\Phi^{t-1}}{nC}
       +A_\eta^{m-t+1}e^{-\eta L}F_\eta^{t-1}
    =\Omega^{t-1}.
\]
Some recipient consequently does not increase $\Omega$.

\emph{Case 2: $Q_i>\tau$ for some agent.}
Use the coefficient of the envy potential at the \emph{new} time $t$, and define
\[
    \lambda:=nC A_\eta^{m-t}e^{-\eta L},\quad
    w_j:=1+\lambda\sum_{i\ne j}e^{\eta E_{ij}^{t-1}}\ge1.
\]
With the exponential factors held fixed, $w_j$ is the coefficient of agent $j$'s PROP1 term in $\Phi+\lambda F_\eta$. Moreover,
\begin{equation}
\label{eq:constant-envy-weighted-sum}
    \sum_j w_jf(s_j^{t-1}+\rho)
    \le\Phi^{t-1}+\lambda F_\eta^{t-1}
    \le nC\Omega^{t-1}<nC.
\end{equation}
The middle inequality uses the larger coefficient $A_\eta\lambda$ at time $t-1$.
Choose a recipient $j$ maximizing $w_jQ_j$ and put $M:=w_jQ_j$. Since some $Q_i>\tau$ and every $w_i\ge1$, we have $M>\tau>0$ and $Q_i\le M$ for every agent. Applying \eqref{eq:constant-envy-large-change} with this common $M$ and then \eqref{eq:constant-envy-weighted-sum} gives
\begin{equation}
\label{eq:constant-envy-total-increase}
    \sum_i w_i d_i
    \le\frac{2001}{2000}\rho M\sum_i w_if(s_i^{t-1}+\rho)
    <\frac{2001}{2000}\frac{19}{30}CM<M.
\end{equation}
Here $(2001/2000)(19/30)C=9998997/10000000<1$.

Assign the good to this $j$. Updating the PROP1 terms while holding the exponentials fixed changes $\Phi+\lambda F_\eta$ by exactly $\sum_iw_id_i-w_j\Delta_j$. Updating the exponentials decreases the terms for agent $j$'s envy toward the other agents. The only positive changes are in the terms for other agents' envy toward $j$, and their sum is at most
\[
    \lambda\sum_{i\ne j}e^{\eta E_{ij}^{t-1}}
       (e^{\eta\widehat v_i(g_t)}-1)(1+f(u_j))
    \le2\eta(w_j-1)(1+f(u_j)).
\]
Consequently,
\begin{align*}
    &(\Phi^t+\lambda F_\eta^t)
        -(\Phi^{t-1}+\lambda F_\eta^{t-1})\\
    &\quad\le\sum_iw_id_i-w_j\Delta_j
                     +2\eta(w_j-1)(1+f(u_j))\\
    &\quad\le\sum_iw_id_i-w_jQ_j<0,
\end{align*}
by \eqref{eq:constant-envy-total-increase}. Dividing by $nC$ and again using the larger envy coefficient at time $t-1$ yields $\Omega^t\le\Omega^{t-1}$.

In both cases a nonincreasing recipient exists. The algorithm chooses a recipient minimizing $\Omega^t$, so it has the same property. This proves the required assertion without requiring a single distribution to preserve every PROP1 term at every time.

\section{A Simpler Simultaneous Guarantee with MIV Predictions} \label{app:miv-combined}

The rule in \Cref{subsec:miv-constant-envy} uses envy terms that depend on the PROP1 slacks. Here we give a simpler rule with separate PROP1 and envy potentials, both of which satisfy a uniform averaging inequality. Its PROP1 factor is smaller, but its normalized envy bound has the explicit constant $11$.
For envy, the form of $\Gamma^t$ and its uniform averaging calculation
follow the exponential potential method of Benad\`e et al.~\cite{benade2018envyvanish}. We choose its parameters so that the initial envy potential is below $1/2$, and then add it to the new PROP1 potential. Their common uniform averaging inequality gives both
guarantees for the same allocation.

Set
\begin{equation}
\label{eq:exp-parameters}
    q:=\frac{2n^2-1}{n-1},
    \quad
    \zeta:=\log q,
    \quad
    \psi:=
    \log\left(
        \frac{(n-1)e^{\zeta/n}+e^{-(n-1)\zeta/n}}{n}
    \right),
\end{equation}
and define
\begin{equation}
\label{eq:exp-alpha}
    \alpha_n
    :=
    1-\frac{n\psi}{\zeta}
    =
    \frac{n}{\log q}\log\left(\frac{q}{2n}\right).
\end{equation}
For every agent with $p_i>0$, let
\begin{equation}
\label{eq:exp-miv-potential}
    \phi_i^t
    :=
    \frac{1}{q}
    \exp\left(
        \zeta\left(\frac{D_i^t}{np_i}-\frac{B_i^t}{p_i}\right)
        -\psi\frac{D_i^t}{p_i}
    \right),
    \quad
    \Phi^t:=\sum_{i:p_i>0}\phi_i^t.
\end{equation}
For a possible recipient $a$, write $\Phi^t(a)$ for the value after assigning the current good to $a$.

\begin{lemma}
\label{lem:exp-miv-potential}
The potential in \eqref{eq:exp-miv-potential} has the following properties:
\begin{enumerate}[(i)]
    \item for every realized history and current good, $\frac1n\sum_{a\in N}\Phi^t(a) \le \Phi^{t-1}$;
    \item $\Phi^0\le1/2$;
    \item if $\phi_i^m\le1$, then agent $i$ satisfies $\alpha_n$-PROP1; and
    \item $0<\alpha_n<1$, and $\frac{1}{2\log(4n)}
    \le
    \alpha_n
    \le
    \frac{2}{\log(2n)}$.
\end{enumerate}
\end{lemma}

\begin{proof}
Fix a time $t$ and an agent $i$ with $p_i>0$. If $t=t_i^*$, then $D_i$ and $B_i$ do not change, so $\phi_i^t(a)=\phi_i^{t-1}$ for every recipient $a$.

Suppose $t\ne t_i^*$ and let $z_i:=v_i(g_t)/p_i\in[0,1]$. If $i$ receives the good, then
\[
    \phi_i^t(i)
    =
    \phi_i^{t-1}
    \exp\left(-\frac{n-1}{n}\zeta z_i-\psi z_i\right).
\]
If another agent receives it, then
\[
    \phi_i^t(a)
    =
    \phi_i^{t-1}
    \exp\left(\frac{\zeta z_i}{n}-\psi z_i\right).
\]
Therefore
\[
\frac1n\sum_{a\in N}\phi_i^t(a)
=
\phi_i^{t-1}e^{-\psi z_i}
\frac{e^{-(n-1)\zeta z_i/n}+(n-1)e^{\zeta z_i/n}}{n}.
\]
Let
\[
    H(z):=
    \log\left(
        \frac{e^{-(n-1)\zeta z/n}+(n-1)e^{\zeta z/n}}{n}
    \right).
\]
The function $H$ is convex, $H(0)=0$, and $H(1)=\psi$. Hence $H(z)\le z\psi$ for $z\in[0,1]$. The average candidate value of $\phi_i$ is therefore at most its current value. Summing over agents proves the first statement.

The identity $(n-1)q+1=2n^2$ gives $e^\psi=2nq^{1/n-1}$.
At time zero, $D_i^0=p_i$ and $B_i^0=0$, so $\phi_i^0
    = \frac1q e^{\zeta/n-\psi}
    = \frac{1}{2n}$.
There are at most $n$ agents with positive $p_i$, and therefore $\Phi^0\le1/2$.

At the final time, $D_i^m=v_i(G)$. The inequality $\phi_i^m\le1$ is equivalent to
\[
    p_i+B_i^m
    \ge
    \left(\frac1n-\frac{\psi}{\zeta}\right)v_i(G)
    =
    \frac{\alpha_n}{n}v_i(G).
\]
The interpretation of $p_i+B_i^m$ following \eqref{eq:miv-accounting} proves $\alpha_n$-PROP1.

For the bounds, let $\frac{q}{2n}=1+z$ and $z:=\frac{2n-1}{2n(n-1)}$.
Then
\[
    \alpha_n=\frac{n\log(1+z)}{\log q}.
\]
Moreover, $0<\alpha_n<1$. Positivity follows from $q>2n$. Also,
\[
    (n-1)\log(1+z)
    <
    (n-1)z
    =
    \frac{2n-1}{2n}
    <
    1
    <
    \log(2n).
\]
Thus $(n-1)\log(1+z)<\log(2n)$, which is equivalent to
$\alpha_n<1$.
For $n\ge2$, we have $z\le1$, $1\le nz\le2$, and $2n<q\le4n$. Using $\frac{z}{1+z}\le\log(1+z)\le z$ gives us the stated bounds.
\end{proof}

The lemma already implies that assigning each good to a recipient minimizing $\Phi^t$ gives $\alpha_n$-PROP1. Its main use here is that the averaging inequality can be combined with the next potential.

\paragraph{Adding the envy potential.}

Let $T=m$ be a known horizon, and assume $T\ge n\log n$. We use the normalized valuations $\widehat v_i$ and the quantities $E_{ij}^t$ defined in \Cref{subsec:miv-constant-envy}.
Set $s:=\sqrt{2\log(1+\frac{n\log n}{T})}$, $C:=1+\frac{e^s+e^{-s}-2}{n}$, and $L:=10\sqrt{\frac{T\log n}{n}}+\frac{\log2}{s}$.
Also define
\begin{equation}
\label{eq:envy-potential}
    \xi_{ij}^t
    :=
    C^{T-t}\exp(s(E_{ij}^t-L)),
    \quad
    \Gamma^t:=\sum_{i\ne j}\xi_{ij}^t.
\end{equation}
For a possible recipient $a$, write $\Gamma^t(a)$ for the value after assigning the current good to $a$.

\begin{lemma}
\label{lem:envy-potential}
The potential in \eqref{eq:envy-potential} has the following properties:
\begin{enumerate}[(i)]
    \item for every realized history and current good, $\frac1n\sum_{a\in N}\Gamma^t(a) \le \Gamma^{t-1}$;
    \item $\Gamma^0<1/2$;
    \item if $\xi_{ij}^T\le1$, then $E_{ij}^T\le L$; and
    \item $L\le11\sqrt{\frac{T\log n}{n}}$.
\end{enumerate}
\end{lemma}

\begin{proof}
Fix $t$, distinct agents $i,j$, and let $w:=\widehat v_i(g_t)\in[0,1]$. Assigning $g_t$ to $j$ increases $E_{ij}$ by $w$, assigning it to $i$ decreases $E_{ij}$ by $w$, and every other assignment leaves it unchanged. Hence
\begin{equation*}
\frac1n\sum_{a\in N}\xi_{ij}^t(a)
= C^{T-t}e^{s(E_{ij}^{t-1}-L)} \frac{e^{sw}+e^{-sw}+n-2}{n} \le C^{T-t}e^{s(E_{ij}^{t-1}-L)} \frac{e^s+e^{-s}+n-2}{n} = \xi_{ij}^{t-1}.
\end{equation*}
Summing over all ordered pairs $(i,j)$ with $i\ne j$ proves the first statement.

At time zero,
\[
    \Gamma^0
    =
    \frac12 n(n-1)C^T
    \exp\left(-10s\sqrt{\frac{T\log n}{n}}\right)
    <
    \frac12\exp\left(
        2\log n+T\log C-10s\sqrt{\frac{T\log n}{n}}
    \right).
\]
Using $\log(1+x)\le x$ and $\cosh x\le e^{x^2/2}$, $T\log C\le2\log n$.
Let $x:=T/(n\log n)\ge1$. Since $x\log(1+1/x)\ge\log2$ for $x\ge1$,
\begin{equation*}
    10s\sqrt{\frac{T\log n}{n}}
    =10\sqrt{2x\log(1+1/x)}\,\log n\ge10\sqrt{2\log2}\,\log n
    >4\log n.
\end{equation*}
Thus $\Gamma^0<1/2$.

At time $T$, $\xi_{ij}^T=\exp(s(E_{ij}^T-L))$, so $\xi_{ij}^T\le1$ implies $E_{ij}^T\le L$.

Finally, $\log(1+u)\ge u/2$ for $u\in[0,1]$ gives us
\[
    s^2=2\log(1+1/x)\ge\frac1x,
    \quad
    \frac1s\le\sqrt{\frac{T}{n\log n}}.
\]
Since $\log2\le\log n$, by the definition of $L$, this proves (iv).
\end{proof}

The recipients minimizing the two potentials separately need not be the same. Their common averaging inequality lets us avoid this issue: we minimize the sum $\Phi^t+\Gamma^t$, whose initial value is below one. Keeping this sum below one makes every term in both potentials less than one at the end, so both guarantees hold for the same allocation.

\paragraph{Joint allocation rule.}
When good $g_t$ arrives, assign it to a recipient that minimizes the combined potential $\Omega^t:=\Phi^t+\Gamma^t$.

\begin{theorem}
\label{thm:miv-prop1-envy}
Assume perfect MIV predictions and a known horizon $T=m\ge n\log n$. Against adaptive adversaries, the joint allocation rule returns an allocation $A$ that simultaneously satisfies the following:
\begin{enumerate}[(i)]
    \item $A$ is $\alpha_n$-PROP1, where $\alpha_n$ is defined in \eqref{eq:exp-alpha} and $\frac{1}{2\log(4n)}
        \le
        \alpha_n
        \le
        \frac{2}{\log(2n)}$;
    \item the maximum additive envy under the normalized valuations satisfies $\max_{i,j\in N}(\widehat v_i(A_j)-\widehat v_i(A_i))
        \le
        11\sqrt{\frac{T\log n}{n}}$. Equivalently, for every $i,j\in N$,
        $v_i(A_j)-v_i(A_i) \le 11p_i\sqrt{\frac{T\log n}{n}}$.
\end{enumerate}
\end{theorem}

\begin{proof}
Both component potentials satisfy the same averaging inequality, so $\frac1n\sum_{a\in N}\Omega^t(a) \le \Omega^{t-1}$.
The joint allocation rule chooses a minimizing recipient, and therefore $\Omega^t\le\Omega^{t-1}$ for every $t$. Initially, $\Omega^0=\Phi^0+\Gamma^0<\frac12+\frac12=1$.
Thus, $\Omega^T<1$. Every term in both potentials is nonnegative, so $\phi_i^T<1$ and $\xi_{ij}^T<1$ for all $i\ne j$. The conclusions follow from \Cref{lem:exp-miv-potential,lem:envy-potential}.
\end{proof}

\end{document}